\documentclass[a4paper,UKenglish,cleveref,autoref,thm-restate]{lipics-v2021}

\hideLIPIcs

\title{Grandchildren-weight-balanced binary search trees}

\author{Vincent Jugé}{LIGM, Univ Gustave Eiffel \& CNRS, France \and IRIF, Université Paris Cité \& CNRS, France}{vincent.juge@univ-eiffel.fr}{https://orcid.org/0000-0003-0834-9082}{}

\authorrunning{V. Jugé}

\Copyright{Vincent Jugé}



\category{} 

\relatedversion{} 

\nolinenumbers 

\EventEditors{Pat Morin and Eunjin Oh}
\EventNoEds{2}
\EventLongTitle{19th International Symposium on Algorithms and Data Structures (WADS 2025)}
\EventShortTitle{WADS 2025}
\EventAcronym{WADS}
\EventYear{2025}
\EventDate{August 11--15, 2025}
\EventLocation{York University, Toronto, Canada}
\EventLogo{}
\SeriesVolume{349}
\ArticleNo{14}

\usepackage{amsmath,amssymb,xspace}
\usepackage[utf8]{inputenc}
\usepackage[T1]{fontenc}
\usepackage{array}
\usepackage{amsthm}
\usepackage{graphicx}
\usepackage{url}
\usepackage{calc}
\usepackage{hyperref}
\usepackage{enumitem}
\usepackage{mathrsfs}
\usepackage[nodisplayskipstretch]{setspace} 
\usepackage{comment}
\usepackage[boxruled,vlined,linesnumbered]{algorithm2e}
\usepackage{graphicx}
\usepackage{xcolor}
\usepackage{colortbl}
\usepackage{thm-restate}
\usepackage{tabularx}
\usepackage{ifthen}
\usepackage{tikz}
\usetikzlibrary{patterns,arrows.meta}
\usepackage{thm-restate}
\usepackage{pifont}
\usepackage{pgfplots}
\usepgfplotslibrary{fillbetween}
\usepackage{marginnote}
\usepackage{multicol}
\usepackage{multirow}

\usepackage{float,xspace}
\usepackage{hyperref}

\usepackage{etoolbox}

\newcommand{\BB}[1]{\mathsf{BB}[#1]}
\newcommand{\GCB}[1]{\mathsf{GCB}[#1]}
\newcommand{\RGCB}[1]{\mathsf{RGCB}[#1]}
\newcommand{\GCBB}[2]{\mathsf{GCB}_{#1}[#2]}
\newcommand{\RGCBB}[2]{\mathsf{RGCB}_{#1}[#2]}
\newcommand{\anc}[2]{#1^{(#2)}}

\newcommand{\balc}{\mathsf{bal_C}}
\newcommand{\balgc}{\mathsf{bal_{GC}}}
\newcommand{\cbal}[1]{\text{\ifthenelse{\equal{#1}{}}{}{$#1$-}\textsf{\small C}-balance\ifthenelse{\equal{#1}{}}{}{d}}}
\newcommand{\gcbal}[1]{\text{\ifthenelse{\equal{#1}{}}{}{$#1$-}\textsf{\small GC}-balance\ifthenelse{\equal{#1}{}}{}{d}}}

\newcommand{\scrB}{\mathscr{B}}

\newcommand{\sA}{\mathsf{C}}

\newcommand{\sB}{\mathsf{GC}}
\newcommand{\sC}{\mathsf{C}}
\newcommand{\scc}{\mathsf{c}}
\newcommand{\sF}{\mathsf{F}}
\newcommand{\sG}{\mathsf{G}}
\newcommand{\sH}{\mathsf{H}}
\newcommand{\sL}{\mathsf{L}}
\newcommand{\bN}{\mathbf{N}}
\newcommand{\sR}{\mathsf{R}}
\newcommand{\sS}{\mathsf{S}}

\newcommand{\D}{\mathcal{D}}
\newcommand{\E}{\mathcal{E}}
\renewcommand{\O}{\mathcal{O}}
\renewcommand{\S}{\mathcal{S}}
\newcommand{\T}{\mathcal{T}}
\newcommand{\U}{\mathcal{U}}

\newcommand{\alphaa}{\hat{\alpha}}
\newcommand{\alphab}{\alpha^\bullet}
\newcommand{\betaa}{\hat{\beta}}

\let\oldnl\nl
\newcommand{\drawline}{\BlankLine\renewcommand{\nl}{\let\nl\oldnl}\hrulefill\BlankLine\setcounter{AlgoLine}{0}}
\SetKwComment{tcp}{ {$\rhd$} }{}
\SetCommentSty{textrm}

\theoremstyle{plain}
\newtheorem{routine}[theorem]{Algorithm}

\newcolumntype{L}[1]{>{\raggedright\let\newline\\\arraybackslash\hspace{0pt}}m{#1}}
\newcolumntype{C}[1]{>{\centering\let\newline\\\arraybackslash\hspace{0pt}}m{#1}}
\newcolumntype{R}[1]{>{\raggedleft\let\newline\\\arraybackslash\hspace{0pt}}m{#1}}

\begin{document}

\maketitle

\begin{abstract}
We revisit weight-balanced trees, also known as trees of bounded balance.
Invented by Nievergelt and Reingold in 1972, these trees are obtained by assigning a weight to each node and requesting that the weight of each node should be quite larger than the weights of its children, the precise meaning of ``quite larger'' depending on a real-valued parameter~$\gamma$.
Blum and Mehlhorn then showed how to maintain them in a recursive (bottom-up) fashion when~$2/11 \leqslant \gamma \leqslant 1-1/\sqrt{2}$, their algorithm requiring only an amortised constant number of tree rebalancing operations per update (insertion or deletion).
Later, in 1993, Lai and Wood proposed a top-down procedure for updating these trees when~$2/11 \leqslant \gamma \leqslant 1/4$.

Our contribution is two-fold.
First, we strengthen the requirements of Nievergelt and Reingold, by also requesting that each node should have a substantially larger weight than its grandchildren, thereby obtaining what we call grandchildren-balanced trees.
Grandchildren-balanced trees are not harder to maintain than weight-balanced trees, but enjoy a smaller node depth, both in the worst case (with a 6~\% decrease) and on average (with a 1.6~\% decrease).
In particular, unlike standard weight-balanced trees, all grandchildren-balanced trees with~$n$ nodes are of height less than~$2 \log_2(n)$.

Second, we adapt the algorithm of Lai and Wood to all weight-balanced trees, i.e., to all parameter values~$\gamma$ such that~$2/11 \leqslant \gamma \leqslant 1-1/\sqrt{2}$.
More precisely, we adapt it to all grandchildren-balanced trees for which~$1/4 < \gamma \leqslant 1 - 1/\sqrt{2}$.
Finally, we show that, except in limit cases (where, for instance,~$\gamma = 1 - 1/\sqrt{2}$), all these algorithms result in making a constant amortised number of tree rebalancing operations per tree update.
\end{abstract}


\section{Introduction}\label{sec:intro}

Among the most fundamental data structures are search trees.
Such trees are aimed at representing sets of pairwise comparable elements of size~$n$ while allowing basic operations such as membership test, insertion and deletion, which typically require a time linear in the depth of the tree.
That is why various data structures such as height-balanced (AVL) trees~\cite{AVL62}, weight-balanced trees~\cite{NR72}, B-trees~\cite{B72}, red-black trees~\cite{RB78}, splay trees~\cite{Splay85}, relaxed~$k$-trees~\cite{FJL01} or weak AVL trees~\cite{HST15} were invented since the 1960s.
All of them provide worst-case or amortized~$\mathcal{O}(\log(n))$ complexities for membership test, insertion and deletion, and may use local modifications of the tree shape called \emph{rotations}.

The most desirable features of such trees include their height, their internal and external path lengths, and the amortised number of rotations triggered by an insertion or deletion.
Such statistics are presented in Table~\ref{table:stats} below for families of binary trees.

Note that, although relaxed~$k$-trees can be made arbitrarily short (given that each tree should be of height at least~$\log_2(n)$), achieving these excellent guarantees on the height of the trees requires performing significantly more rotations, and thus considering other data structures may still be relevant.
This is, for instance, why Haeupler, Tarjan and Sen proposed weak AVL trees in 2009: these are a variant of AVL trees, isomorphic to red-black trees, whose height can be worse than that of AVL trees (but never worse than that of red-black trees), but which require only a constant number of rotations per update, whereas standard AVL trees may require up to~$\log(n)$ rotations per update.

{\renewcommand*{\arraystretch}{1.1}
\setlength{\tabcolsep}{2pt}
\begin{table}[t]
\begin{center}
\begin{tabular}{|C{26mm}|L{0mm}C{29mm}R{0mm}|L{0mm}C{34mm}R{0mm}|L{0mm}C{30mm}R{0mm}|}
\hline
\multirow{2}{*}{Tree family} & \multicolumn{9}{c|}{Worst-case} \\
\cline{2-10}
 & \multicolumn{3}{c|}{height} & \multicolumn{3}{c|}{path length} & \multicolumn{3}{c|}{am. rotations/update} \\
\hline
AVL & & {}$1.4404 \log_2(n)$& \llap{\cite{AVL62}} & & {}$1.4404 n \log_2(n)$ & \llap{\cite{KW90}} & & {}$\Theta(\log(n))$ & \llap{\cite{ALT16}} \\
\hline
Weight-balanced & & {}$2 \log_2(n)$ & \llap{\cite{NR72}} & & {}$1.1462 n \log_2(n)$ & \llap{\cite{NR72}} & & {}$\Theta(1)$ & \llap{\cite{BM80}} \\
\hline
Red-black & & {}$2 \log_2(n)$ & \llap{\cite{RB78}} & & {}$2 n \log_2(n)$ & \llap{\cite{CW92}} & & {}$\Theta(1)$ & \llap{\cite{RB78}} \\
\hline
Splay & & {}$n$ & \llap{\cite{Splay85}} & & {}$n^2/2$ & \llap{\cite{Splay85}} & & {}$\Theta(1)$ & \llap{\cite{Splay85}} \\
\hline
Relaxed~$k$- & & {}$(1 + \varepsilon) \log_2(n)$ & \llap{\cite{FJL01}} & & {}$(1 + \varepsilon) n \log_2(n)$ & \llap{\cite{FJL01}} & & {}$\O(1 / \varepsilon)$& \llap{\cite{FJL01}} \\
\hline
Weak AVL & & {}$2 \log_2(n)$& \llap{\cite{RB78}} & & {}$2 n \log_2(n)$ & \llap{\cite{CW92}} & & {}$\Theta(1)$ & \llap{\cite{HST15}} \\
\hline
\hline
Grandchildren-balanced & & {}$1.8798 \log_2(n)$ & & & {}$1.1271 n \log_2(n)$ & & & {}$\Theta(1)$ & \\
\hline
\end{tabular}
\end{center}
\caption{Asymptotic approximate worst-case height, internal/external path length and amortised number of rotations per update in binary trees with~$n$ nodes.
Next to each worst-case bound are indicated references where this bound is proved.}
\label{table:stats}
\end{table}
}

In 1972, Nievergelt and Reingold invented \emph{weight-balanced trees}, or \emph{trees of bounded balance}~\cite{NR72}.
This family of trees depends on a real parameter~$\gamma$, and is denoted by~$\BB{\gamma}$.
It is based on the notion of \emph{weight} of a node~$n$, which is the integer~$|n|$ defined as one plus the number of descendants of a node~$n$;
alternatively,~$|n|$ is the number of empty sub-trees descending from~$n$.
Although they might be less efficient than other families of balanced trees in general, they are a reference data structure for order-statistic trees: they are relevant in contexts where we need to store the rank of elements in a dynamic set~\cite{BJ09} (i.e., finding efficiently the~$k$\textsuperscript{th} smallest element, or counting elements smaller than a given bound).
They are also widely used in libraries of mainstream languages, e.g., Haskell set or map implementations~\cite{Haskell10a,Haskell10b}.

The family~$\BB{\gamma}$ consists of those binary search trees in which, for all nodes~$u$ and~$v$, we have~$|u| \geqslant \gamma |v|$ whenever~$u$ is a child of~$v$.
Nievergelt and Reingold also gave an algorithm for dynamically maintaining~$\BB{\gamma}$-trees of size~$n$ whenever~$\gamma \leqslant \sqrt{2}-1$, which required only~$\O(\log(n))$ element comparisons and pointer moves per modification.

Blum and Mehlhorn then found that this algorithm worked only when~$2/11 \leqslant \gamma \leqslant \sqrt{2}-1$.
They also proved that, in that case, it required only~$\O(1)$ amortised pointer moves per modification~\cite{BM80}.
Lai and Wood subsequently proposed an algorithm for maintaining such trees through a single top-down pass per modification whenever~$2/11 \leqslant \gamma \leqslant 1/4$, or whenever updates were guaranteed to be non-redundant~\cite{LW93}.
This limitation on the domain of~$\gamma$ when tackling possibly redundant updates is somewhat unfortunate, given that those values of~$\gamma$ for which~$\BB{\gamma}$ have the best upper bounds on their height and path length are those for which~$\gamma$ approaches~$\sqrt{2}-1$.

Top-down maintenance algorithms have two main advantages over bottom-up algorithms.
First, bottom-up algorithms require maintaining a link from each node to its parent or storing the list of nodes we visited along a given branch on a stack, which makes them typically slower than top-down algorithms~\cite{BW20}.
Second, they are also a strong bottleneck in concurrent or parallel settings: going down to a leaf and then back to the root must be an atomic operation (or requires maintaining children-to-parent links), since it might finish by changing the root~\cite{RB78,HST15}.

Thus, our goal here is two-fold.
First, by making small changes to a long-known data structure, we improve the guarantees it offers in terms of tree height.
Second, we aim at circumventing the limitations of the algorithm from Lai and Wood, by proposing a top-down updating algorithm that will be valid in all cases, and which we also extend to our enhanced data structure.

\paragraph*{Contributions}

We propose a new variant of weight-bounded trees, which we call \emph{grandchildren-balanced trees}.
This family of trees depends on two real parameters~$\alpha$ and~$\beta$, and is denoted by~$\GCB{\alpha,\beta}$.
It consists of those binary search trees in which, for all nodes~$u$ and~$v$, we have~$|u| \leqslant \alpha |v|$ whenever~$u$ is a child of~$v$, and~$|u| \leqslant \beta |v|$ whenever~$u$ is a grandchild of~$v$.
Grandchildren-balanced trees generalise weight-bounded trees, because~$\GCB{\alpha,\alpha^2} = \BB{1-\alpha}$.

We prove that, for well-chosen values of~$\alpha$ and~$\beta$, the height and internal and external path lengths of grandchildren-balanced trees are smaller than those of weight-bounded trees;
in particular, we break the~$2 \log_2(n)$ lower bound on the worst-case height of weight-balanced trees.
These results are achieved by using Algorithm~\ref{alg:BU-rebalancing}.

We also extend the algorithm of Lai and Wood, and propose an algorithm for maintaining grandchildren-balanced trees in a single top-down pass per modification, which requires only~$\O(1)$ amortised pointer moves per modification; this is Algorithm~\ref{alg:TD-updating}.
Our algorithm is valid for all relevant values of~$\alpha$ and~$\beta$, including those for which~$\GCB{\alpha,\beta}$-trees enjoy the best upper bounds on their height and path length.

Several proofs are omitted, or only sketched, in the body of this article;
they can be found in the appendix.

\section{Grandchildren-balanced trees}\label{sec:grandchildren}

In this section, we describe a new family of binary search trees, called \emph{grandchildren-balanced trees}, which depends on two real parameters~$\alpha$ and~$\beta$.
We also describe the domain in which the pair~$(\alpha,\beta)$ is meaningful and, for such pairs, we provide upper bounds on the height and internal and external path lengths of grandchildren-balanced trees.

Let us first recall the definition given in Section~\ref{sec:intro}.

\begin{definition}\label{def:GB-trees}
The \emph{weight} of a binary tree~$\T$, denoted by~$|\T|$, is defined as the number of empty sub-trees of~$\T$; alternatively,~$|\T|-1$ is the number of nodes of~$\T$.
The \emph{weight} of a node~$n$ of~$\T$, also denoted by~$|n|$, is defined as the weight of the sub-tree of~$\T$ rooted at~$n$.

In what follows, it may be convenient to see each empty tree (of weight~$1$) as if it had two empty children of weight~$1/2$ each, themselves having two empty children of weight~$1/4$ each, and so on.
In other words, we may see each empty tree as an infinite complete binary tree whose nodes of depth~$d$ have weight~$2^{-d}$.

Let~$n$ be a node of~$\T$, let~$n_1$ and~$n_2$ be its children, and let~$n_{11}$,~$n_{12}$,~$n_{21}$ and~$n_{22}$ be its grandchildren.
The \emph{child-balance} and the \emph{grandchild-balance} (also called \cbal{} and \gcbal{}) of~$n$ in~$\T$ are defined as the real numbers
\[\balc(\T,n) = \frac{\max\{|n_1|,|n_2|\}}{|n|} \text{ and } \balgc(\T,n) = \frac{\max\{|n_{11}|,|n_{12}|,|n_{21}|,|n_{22}|\}}{|n|}.\]
When the context is clear, we may omit referring to~$\T$, thereby simply writing~$\balc(n)$ and~$\balgc(n)$;
alternatively, when~$n$ is the root of~$\T$, its \cbal{} and the \gcbal{} may also be directly denoted by~$\balc(\T)$ and~$\balgc(\T)$.

Given real numbers~$\alpha$ and~$\beta$, we say that a node~$n$ of~$\T$ is \emph{$\alpha$-child balanced} (or~$\cbal{\alpha}$) when~$\balc(n) \leqslant \alpha$;
\emph{$\beta$-grandchild-balanced} (or~$\gcbal{\beta}$) when~$\balgc(n) \leqslant \beta$;
and~\emph{$(\alpha,\beta)$-balanced} when~$n$ is both~$\cbal{\alpha}$ and~$\gcbal{\beta}$.

Finally, we say that~$\T$ is a~$\GCBB{x,y}{\alpha,\beta}$-tree when its root is~$(x,y)$-balanced and its other nodes are~$(\alpha,\beta)$-balanced.
For the sake of concision, we simply say that~$\T$ is a~$\GCBB{x}{\alpha}$-tree, a~$\GCB{\alpha,\beta}$-tree or a~$\GCB{\alpha}$-tree when~$(y,\beta) = (1,1)$,~$(x,y) = (\alpha,\beta)$ or~$(x,y,\beta) = (\alpha,1,1)$, respectively.
\end{definition}

Below, we will be mostly interested in the family of~$\GCB{\alpha,\beta}$-trees.
Although it generalises weight-bounded trees invented by Nievergelt and Reingold, we had to shift away from their notation.
Indeed, weight-bounded trees were parametrised with a real number~$\gamma$, by requesting that~$|u| \geqslant \gamma |v|$ whenever~$u$ is a child of~$v$;
however, being~$\gcbal{\beta}$ cannot be expressed by using constraints such as ``$|u| \geqslant \gamma |v|$ whenever~$u$ is a grandchild of~$v$''.
In particular, in~\cite{BM80,NR72}, the \emph{root-balance} of the node~$n$ is simply defined as the real number~$\rho(n) = |n_1|/|n|$;
here, we set~$\balc(n) = \max\{\rho(n),1-\rho(n)\}$.

The family of~$\GCB{\alpha,\beta}$-trees coincides with
\begin{itemize}[itemsep=0.5pt,topsep=0.5pt]
\item both~families of~$\GCB{\alpha}$-trees and~$\BB{1-\alpha}$-trees when~$\alpha^2 \leqslant \beta$;
\item the empty set when~$\alpha < 1/2$;
\item the family of~$\GCB{1,\beta}$-trees when~$1 \leqslant \alpha$;
\item the family of~$\GCB{2\beta,\beta}$-trees when~$\beta < \alpha/2$;
\item the family of all binary trees when~$\alpha = \beta = 1$.
\end{itemize}
Thus, we assume below that the pair~$(\alpha,\beta)$ belongs to the domain
\[\D = \{(\alpha,\beta) \colon \alpha/2 \leqslant \beta \leqslant \alpha^2 \text{ and } 1/2 \leqslant \alpha \leqslant 1\} \setminus \{(1,1)\},\]
represented in Figure~\ref{fig:domains} along with two other domains~$\D'$ and~$\D''$ at the end of Section~\ref{sec:grandchildren}.

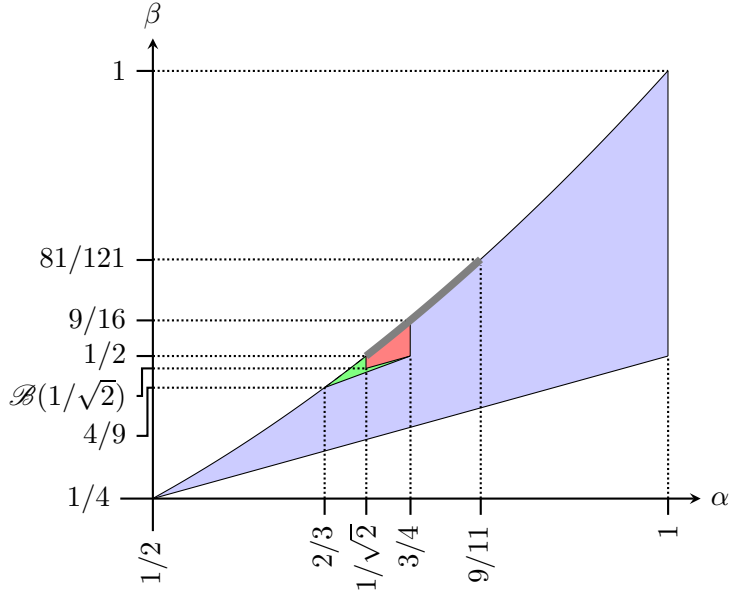
\begin{figure}[t]
\begin{center}
\begin{tikzpicture}[scale=1]
\begin{axis}[thick,smooth,no markers,hide axis]
\addplot+[thin,name path=A,black,domain=0.5:1] {x/2};
\addplot+[thin,name path=B,black,domain=0.5:1] {x*x};

\addplot+[thin,name path=C,black,domain=0.70711:0.75] {0.5*sqrt(1+4*x)-0.5};
\addplot+[very thin,name path=D,black,domain=0.70711:0.75] {x*x};

\addplot+[thin,name path=E,black,domain=0.66667:0.75] {0.66667*x};
\addplot+[very thin,name path=F,black,domain=0.66667:0.75] {x*x};

\addplot[blue!20] fill between[of=A and B];
\addplot[green!50] fill between[of=E and F];
\addplot[red!50] fill between[of=C and D];
\addplot+[thin,black,solid,domain=0.5:1] (1,x);
\addplot+[thin,black,solid,domain=0.47832:0.5] (0.70711,x);
\addplot+[thin,black,solid,domain=0.5:0.5625] (0.75,x);
\addplot+[line width=1mm,black!50,domain=0.70711:0.818182] {x*x};
\end{axis}
\begin{scope}[shift={(-5.125,-1.11)},xscale=11.416,yscale=6.322,>=stealth]
\draw[thick,->] (0.46839,0.25) -- (1.03161,0.25);
\draw[thick,->] (0.5,0.19292) -- (0.5,1.05708);
\draw[thick,densely dotted] (0.5,0.44444) -- (0.66667,0.44444);
\draw[thick,densely dotted] (0.5,0.4783) -- (0.70711,0.4783);
\draw[thick,densely dotted] (0.5,0.5) -- (0.70711,0.5);
\draw[thick,densely dotted] (0.5,0.5625) -- (0.745,0.5625);
\draw[thick,densely dotted] (0.5,0.669421) -- (0.813182,0.669421);
\draw[thick,densely dotted] (0.5,1) -- (1,1);

\draw[thick,densely dotted] (0.66667,0.25) -- (0.66667,0.44444);
\draw[thick,densely dotted] (0.70711,0.25) -- (0.70711,0.4783);
\draw[thick,densely dotted] (0.75,0.25) -- (0.75,0.5);
\draw[thick,densely dotted] (0.818182,0.25) -- (0.818182,0.659421);
\draw[thick,densely dotted] (1,0.25) -- (1,0.5);

\node[anchor=west] at (1.03161,0.25) {$\alpha$};
\node[anchor=south] at (0.5,1.05708) {$\beta$};
\node[anchor=east,rotate=90] at (0.5,0.19292) {$1/2$};
\node[anchor=east] at (0.46839,0.25) {$1/4$};

\foreach \y/\z/\l/\p in {0.44444/0.36/{4/9}/0.33333,0.4783/0.43/{\scrB(1/\sqrt{2})}/0.66667,0.5/0.5/{1/2}/1,0.5625/0.5625/{9/16}/1,0.668421/0.668421/{81/121}/1,1/1/1/1}{
\draw[thick] (0.5,\y) -- (0.5-\p*0.015805,\y) -- (0.5-\p*0.015805,\z) -- (0.5-0.015805,\z);
\node[anchor=east] at (0.5-0.015805,\z) {$\l$};
}

\foreach \x/\l in {0.66667/{2/3},0.70711/{1/\sqrt{2}},0.75/{3/4},0.818182/{9/11},1/1}{
\draw[thick] (\x,0.25) -- (\x,0.22146);
\node[anchor=east,rotate=90] at (\x,0.22146) {$\l$};
}
\end{scope}
\end{tikzpicture}
\vspace{-1.2em}
\end{center}
\caption{Two-dimensional domains~${\color{red}\D'} \subseteq {\color{green}\D''} \subseteq {\color{blue!50}\D}$, and one-dimensional (thick, gray) curve~${\color{black!50}\mathcal{C}}$.
Whereas~$\BB{1-\alpha}$-trees correspond to choosing~$(\alpha,\beta)$ on the curve~${\color{black!50}\mathcal{C}}$, our~$\GCB{\alpha,\beta}$-trees can be created on the entire domain~${\color{red}\D'}$.
\label{fig:domains}}
\end{figure}

We first obtain the following upper bound on the height of~$\GCB{\alpha,\beta}$-trees.

\begin{theorem}\label{thm:height}
When~$(\alpha,\beta) \in \D$, each~$\GCB{\alpha,\beta}$-tree with~$\bN$ nodes is of height
\[h \leqslant - 2 \log_\beta(\bN+1).\]
\end{theorem}

\begin{proof}
Let~$\T$ be a~$\GCB{\alpha,\beta}$-tree with~$\bN$ nodes.
Let~$h$ be its height, and let~$n$ be a node of~$\T$ at depth~$h$.
An induction on~$k$ proves that~$|\anc{n}{k}| \geqslant 2 \beta^{-k/2} \alpha$ whenever~$0 \leqslant k \leqslant h$, 
where~$\anc{n}{k}$ denotes the~$k$\textsuperscript{th} ancestor of~$n$, and is defined by~$n$ if~$k = 0$, or as the parent of~$\anc{n}{k-1}$ if~$k \geqslant 1$.
It follows that~$\bN + 1 \geqslant |\anc{n}{h}| \geqslant 2 \beta^{-h/2} \alpha \geqslant \beta^{-h/2}$, i.e., that~$h \leqslant - 2 \log_\beta(\bN+1)$.
\end{proof}

We also focus on the internal and external path lengths of~$\GCB{\alpha,\beta}$-trees.
The \emph{internal} path length of a tree~$\T$ is the sum of the lengths of the paths from the root of~$\T$ to the nodes of~$\T$;
the \emph{external} path length of~$\T$ is the sum of the lengths of the paths from the root of~$\T$ to the empty sub-trees of~$\T$.
These lengths are closely related to the average number of queries required to check membership in the set of labels of~$\T$.

Indeed, let~$\T$ be a binary search tree whose~$\bN$ nodes are labelled by elements of a linearly ordered set~$\E$, and let~$\lambda_{\mathsf{int}}$ and~$\lambda_{\mathsf{ext}}$ be its internal and external path lengths.
The labels of the nodes of~$\T$ form a linearly ordered set~$\S \subseteq \E$ of size~$\bN$, which splits~$\E$ into~$\bN+1$ intervals.
Then, the average number of queries used to find an element~$x$ is:
\begin{itemize}[itemsep=0.5pt,topsep=0.5pt]
\item {}$\lambda_{\mathsf{int}}/\bN$ when~$x$ is chosen uniformly at random in~$\S$;
\item {}$\lambda_{\mathsf{ext}}/(\bN+1)$ when~$x$ is chosen in~$\E \setminus \S$, and the interval of~$\E \setminus \S$ to which~$x$ belongs is chosen uniformly at random.
\end{itemize}

By construction, each node~$n$ increases the lengths~$\lambda_{\mathsf{int}}$ and~$\lambda_{\mathsf{ext}}$ by~$|n|-2$ and~$|n|$, respectively.
This means that~$\lambda_{\mathsf{ext}}$ is the sum of the weights of the tree nodes, and that~$\lambda_{\mathsf{int}} = \lambda_{\mathsf{ext}} - 2 \bN$.

By adapting the proof from~\cite{NW73}, we obtain the following upper bound on the external path lengths of~$\GCB{\alpha,\beta}$-trees.

\begin{restatable}{theorem}{ipl}\label{thm:IPL}
When~$(\alpha,\beta) \in \D$, each~$\GCB{\alpha,\beta}$-tree with~$\bN$ nodes is of external path length
\[\lambda \leqslant (\bN+1) \log_2(\bN+1) / \Delta,\]
where~$\Delta = (\sH_2(\alpha) + \alpha \sH_2(\beta/\alpha))/(1+\alpha)$, and~$\sH_2(x) = -x \log_2(x) - (1-x) \log_2(1-x)$ is Shannon's binary entropy.
\end{restatable}

In what follows, we will investigate more closely the family of~$\GCB{\alpha,\beta}$-trees when~$(\alpha,\beta)$ belongs to the more restricted domain
\[\D' = \{(\alpha,\beta) \colon 1/\sqrt{2} \leqslant \alpha < 3/4 \text{ and } \scrB(\alpha) \leqslant \beta \leqslant \alpha^2\} \text{, where } \scrB(\alpha) = \frac{\sqrt{1+4\alpha}-1}{2}.\]
In the limit cases where~$\alpha = 1/\sqrt{2} \approx 0.7071$ and~$\beta = \scrB(\alpha) \approx 0.4783$, Theorems~\ref{thm:height} and~\ref{thm:IPL} translate into the inequalities
\[h \leqslant 1.8798 \log_2(\bN+1) \text{ and }
\lambda \leqslant 1.1271 (\bN+1) \log_2(\bN+1),\]
thereby improving the inequalities
\[h \leqslant 2 \log_2(\bN+1) \text{ and }
\lambda \leqslant 1.1462 (\bN+1) \log_2(\bN+1)\]
obtained in the best case (i.e., in the limit case where~$\alpha = 1/\sqrt{2}$ and~$\beta = 1/2$) for~$\BB{1-\alpha}$-trees.

On the domain~$\D'$, or even on the larger domain~$\D'' = \{(\alpha,\beta) \colon 2/3 \leqslant \beta / \alpha \leqslant \alpha < 3/4\}$, the upper bounds presented in Theorem~\ref{thm:height} and~\ref{thm:IPL} are (unsurprisingly) quite tight.

\begin{restatable}{proposition}{tightipl}
Let~$(\alpha,\beta)$ be a pair belonging to~$\D''$, and let~$\bN \geqslant 0$ be an integer.
There exists a~$\GCB{\alpha,\beta}$-tree with~$\bN$ nodes, of height~$h \geqslant -2 \log_2(\bN+1)/\log_2(\beta) - 7$ and external path length~$\lambda \geqslant (\bN+1) \log_2(\bN+1)/\Delta - 4 \bN$.
\end{restatable}

\section{Bottom-up algorithm}\label{sec:BU}

In this section, we propose a simple algorithm for rebalancing a~$\GCB{\alpha,\beta}$ tree~$\T$ after having added a leaf to~$\T$ or deleted a node from~$\T$, when the pair~$(\alpha,\beta)$ belongs to the domain~$\D'$.

Rebalancing~$\GCB{\alpha,\beta}$ trees might also be possible when~$3/4 \leqslant \alpha \leqslant 1$ and~$\mathscr{B}(\alpha) \leqslant \beta \leqslant \alpha^2$ by using similar ideas, but doing so may require more complicated algorithms, in particular for dealing with base cases, which is not necessarily worth the effort since such parameters provide us with trees of larger height and path length.

In case of a deletion, as illustrated in the three cases of Figure~\ref{fig:in-out}, we can safely assume that the node to be deleted is a leaf~$s$, by using Hibbard's deletion technique~\cite{H62}.
Indeed, if the targeted node~$a$ has just one child, since~$a$ is~$\cbal{\alpha}$ and~$\alpha < 3/4$, we know that~$|a| \leqslant 3$, which means that this child is a leaf; thus, we can just focus on deleting that child and, \emph{a posteriori}, replace the key of~$a$ with the key of this just-deleted child.
Similarly, if the targeted node~$a$ has two children, we can first identify the successor of~$a$ as the deepest node~$b$ of the left branch stemming from the right child of~$a$;
$b$ has no left child, which means that we can just focus on deleting it and, \emph{a posteriori}, replace the key of~$a$ with the key of this just-deleted node~$b$.

In both cases, let~$r$ be the root of~$\T$, and let~$s$ be the node that must be inserted or deleted.
The weights of the nodes from~$r$ to~$s$ need to be incremented by~$1$ in case of a non-redundant insertion (since~$s$ just replaced an empty tree, we can pretend that its weight was~$1$ before it was inserted), or decremented by~$1$ in case of a non-redundant deletion (since~$s$ was a leaf and is replaced by an empty tree, we can pretend that its new weight is~$1$).
However, if the update is redundant, i.e., if we tried to insert a key that was already present in~$\T$, or to delete an absent key, these weights will not be changed.

In practice, each tree node will contain two fields for storing explicitly its label and its weight.
Moreover, in the description below, we simply say ``\,if~$s$ lies in a sub-tree~$\T'$\,'' as a place-holder for ``\,determine, based on the key~$x$ that you want to insert or delete and on the key~$y$ stored in the root of the current tree, which sub-tree~$\T'$ should contain~$s$\,''.

\newcommand{\drawkey}[4]{
\begin{scope}[shift={(#1-#4*1.55,#2+0.65)},xscale=#4*0.5,yscale=0.5]
\draw[fill=white] (0,0.2) -- (-2,0.2) -- (-2.2,0.4) -- (-1.8,0.8) -- (-1.6,0.6) -- (-1.4,0.8) -- (-1.2,0.6) -- (-1.05,0.75) -- (-0.95,0.75) -- (-0.8,0.6) -- (-0.6,0.8) -- (0,0.8);
\draw[fill=white] (0,0) --++ (0,1) --++ (0.7,0.7) --++ (1,0) --++ (0.7,-0.7) --++ (0,-1) --++ (-0.7,-0.7) --++ (-1,0) -- cycle;
\draw[fill=black] (1.825,0.5) circle (0.25);
\draw[very thick] (1.825,0.5) -- (2.6,0.5) arc (90:0:0.5) --++ (0,-0.2);
\node at (0.8,0.5) {#3};
\end{scope}
}

\begin{figure}[t]
\begin{center}
\begin{tikzpicture}[scale=0.44]
\draw[draw=gray,fill=gray!30] (0.4,7.8) -- (4.5,3.7) -- (4.5,0) -- (0.2,0)  -- (-0.6,-0.8) -- (-1.4,-0.8) -- (-2.2,0) -- (-4.5,0) -- (-4.5,3.7) -- (-0.4,7.8) -- cycle;
\draw[ultra thick] (0,7) -- (1,6) -- (0,5) -- (2,3) -- (1,2) -- (-1,0);
\draw[ultra thick,draw=gray] (-2,-1) -- (-1,0);
\draw[fill=white,draw=white] (-2,-1) circle (0.2);
\node at (-2,-1) {$\llap{$s =\,$}\bot$};
\foreach \x/\y/\l in {0/7/r,-1/0/{}}{
 \draw[fill=white,thick] (\x,\y) circle (0.55);
 \node at (\x,\y) {$\l$};
}
\draw[line width=3pt,->,>=stealth] (0,-2) -- (0,-4);

\begin{scope}[shift={(0,-12)}]
\draw[draw=gray,fill=gray!30] (0.4,7.8) -- (4.5,3.7) -- (4.5,0) -- (0.2,0) -- (-1.6,-1.8) -- (-2.4,-1.8) -- (-2.8,-1.4) -- (-2.8,0) -- (-4.5,0) -- (-4.5,3.7) -- (-0.4,7.8) -- cycle;
\draw[ultra thick] (0,7) -- (1,6) -- (0,5) -- (2,3) -- (1,2) -- (-2,-1);
\foreach \x in {-3,-1}{
 \draw[ultra thick,draw=gray] (\x,-2) -- (-2,-1);
 \draw[fill=white,draw=white] (\x,-2) circle (0.2);
 \node at (\x,-2) {$\bot$};
}
\foreach \x/\y/\l in {0/7/r,-2/-1/s,-1/0/{}}{
 \draw[fill=white,thick] (\x,\y) circle (0.55);
 \node at (\x,\y) {$\l$};
}
\drawkey{-2}{-1}{$x$}{1}
\end{scope}

\draw[ultra thick,densely dashed] (5.5,-14.3) --++ (0,22.6);

\begin{scope}[shift={(11,0)}]
\draw[draw=gray,fill=gray!30] (0.4,7.8) -- (4.5,3.7) -- (4.5,0) -- (0.2,0) -- (-1.6,-1.8) -- (-2.4,-1.8) -- (-2.8,-1.4) -- (-2.8,0) -- (-4.5,0) -- (-4.5,3.7) -- (-0.4,7.8) -- cycle;
\draw[ultra thick] (0,7) -- (1,6) -- (0,5) -- (2,3) -- (-2,-1);
\foreach \x in {-3,-1}{
 \draw[ultra thick,draw=gray] (\x,-2) -- (-2,-1);
 \draw[fill=white,draw=white] (\x,-2) circle (0.2);
 \node at (\x,-2) {$\bot$};
}
\foreach \x/\y/\l in {0/7/r,-1/0/{},1/4/a,2/3/{},-2/-1/s}{
 \draw[fill=white,thick] (\x,\y) circle (0.55);
 \node at (\x,\y) {$\l$};
}
\drawkey{1}{4}{$x$}{-1}
\drawkey{-2}{-1}{$y$}{1}

\draw[line width=3pt,->,>=stealth] (0,-2) -- (0,-4);
\end{scope}

\begin{scope}[shift={(11,-12)}]
\draw[draw=gray,fill=gray!30] (0.4,7.8) -- (4.5,3.7) -- (4.5,0) -- (0.2,0) -- (-0.6,-0.8) -- (-1.4,-0.8) -- (-2.2,0) -- (-4.5,0) -- (-4.5,3.7) -- (-0.4,7.8) -- cycle;
\draw[ultra thick] (0,7) -- (1,6) -- (0,5) -- (2,3) -- (-1,0);
\foreach \x in {-2,0}{
 \draw[ultra thick,draw=gray] (\x,-1) -- (-1,-0);
 \draw[fill=white,draw=white] (\x,-1) circle (0.2);
}
\node at (0,-1) {$\bot$};
\node at (-2,-1) {$\llap{$s =\,$}\bot$};
\foreach \x/\y/\l in {0/7/r,-1/0/{},1/4/a,2/3/{},-1/0/{}}{
 \draw[fill=white,thick] (\x,\y) circle (0.55);
 \node at (\x,\y) {$\l$};
}
\drawkey{1}{4}{$y$}{-1}
\end{scope}

\draw[ultra thick,densely dashed] (16.5,-14.3) --++ (0,22.6);

\begin{scope}[shift={(22,0)}]
\draw[draw=gray,fill=gray!30] (0.4,7.8) -- (4.5,3.7) -- (4.5,0) -- (0.8,0) -- (0.8,-1.4) -- (0.4,-1.8) -- (-0.4,-1.8) -- (-2.2,0) -- (-4.5,0) -- (-4.5,3.7) -- (-0.4,7.8) -- cycle;
\draw[ultra thick] (0,7) -- (1,6) -- (0,5) -- (2,3) -- (-1,0) -- (0,-1);
\draw[ultra thick,draw=gray] (-2,-1) -- (-1,0) (-1,-2) -- (0,-1) -- (1,-2);
\foreach \x/\y in {-2/-1,-1/-2,1/-2}{
 \draw[fill=white,draw=white] (\x,\y) circle (0.2);
 \node at (\x,\y) {$\bot$};
}
\foreach \x/\y/\l in {0/7/r,-1/0/{},0/-1/s,1/4/a,2/3/{},-1/0/b}{
 \draw[fill=white,thick] (\x,\y) circle (0.55);
 \node at (\x,\y) {$\l$};
}
\drawkey{1}{4}{$x$}{-1}
\drawkey{-1}{0}{$y$}{1}
\drawkey{0}{-1}{$z$}{-1}
\draw[line width=3pt,->,>=stealth] (0,-2) -- (0,-4);
\end{scope}

\begin{scope}[shift={(22,-12)}]
\draw[draw=gray,fill=gray!30] (0.4,7.8) -- (4.5,3.7) -- (4.5,0) -- (0.2,0) -- (-0.6,-0.8) -- (-1.4,-0.8) -- (-2.2,0) -- (-4.5,0) -- (-4.5,3.7) -- (-0.4,7.8) -- cycle;
\draw[ultra thick] (0,7) -- (1,6) -- (0,5) -- (2,3) -- (1,2) -- (-1,0);
\foreach \x in {-2,0}{
 \draw[ultra thick,draw=gray] (\x,-1) -- (-1,-0);
 \draw[fill=white,draw=white] (\x,-1) circle (0.2);
}
\node at (-2,-1) {$\bot$};
\node at (0,-1) {$\bot\rlap{$\,= s$}$};
\foreach \x/\y/\l in {0/7/r,-1/0/{},1/4/a,2/3/{},-1/0/b}{
 \draw[fill=white,thick] (\x,\y) circle (0.55);
 \node at (\x,\y) {$\l$};
}
\drawkey{1}{4}{$y$}{-1}
\drawkey{-1}{0}{$z$}{1}
\end{scope}
\end{tikzpicture}
\vspace{-0.6em}
\end{center}
\caption{Inserting or deleting a key~$x$: either an empty sub-tree grows into a leaf or a leaf shrinks to an empty sub-tree.
When we wish to delete an internal node, we just focus on deleting a leaf descending from that node, and then swap node keys.
\label{fig:in-out}}
\end{figure}

Our algorithm is based on two algorithmic building blocks, called~$\sA$-balancing (Algorithm~\ref{routine:A}) and~$\sB$-balancing (Algorithm~\ref{routine:B}).
They consist in locally performing rotations in order to make a given node better balanced when needed, without damaging the balance of its parent and children too much.

In a nutshell, the idea of our bottom-up updating algorithm, which will be made more precise in Algorithm~\ref{alg:BU-rebalancing}, is as follows:
\begin{enumerate}[itemsep=0.5pt,topsep=0.5pt]
\item We recursively update the sub-tree that needs to be updated.
\item Using the~$\sA$-balancing algorithm ensures our root and its children are suitably child-balanced.
\item Using the~$\sB$-balancing algorithm then ensures our root is suitably child- and grandchild-balanced.
\end{enumerate}
This is the same general idea as the original algorithm from Nievergelt and Reingold~\cite{BM80,NR72}, adapted to take grandchild balances into account: one should first make our tree child-balanced, and only then can one make it grandchild-balanced too.

One difficulty is that, when we perform a rotation to make our root child- or grandchild-balanced, this rotation should be worth its cost: we expect that, for the next few updates (a quantity that can be made precise, and will be studied in Section~\ref{sec:complexity}), no update will be required at all.
This is what will make the amortised number of rotations per update a constant.

In what follows,~$\T$ denotes a binary tree that may need to be rebalanced.
Its root is denoted by~$n$ and, for every node~$x$, the left and right children of~$x$ are denoted by~$x_1$ and~$x_2$, respectively.
Furthermore, we consider that nodes \emph{move} when rotations are performed.
For instance, when either rotation represented in Figure~\ref{fig:balancing} is performed, the node~$n$, which used to be the root of the tree, becomes the right child of the root of the tree resulting from the rotation.

Below, we also say that a node whose children list changed during the rotation was \emph{affected} by the rotation.
We will always make sure that, when a simple or double rotation is performed on a tree~$\T$, the resulting tree~$\T'$ satisfies the inequality~$\balc(\T') \leqslant \balc(\T)$.
That way, if~$\T$ was in fact rooted at some child of a node~$x$, neither the~$\cbal{}$ nor the~$\gcbal{}$ of~$x$, or of any unaffected node, will increase as a result of the rotation.

\begin{figure}[t]
\begin{center}
\begin{tikzpicture}[scale=0.85]
\foreach \x/\y/\c in {0/0/11,1/1/1,2/0/12,2/2/,3/1/2,1/-1/121,3/-1/122}{
 \coordinate (n\c) at (0.75*\x,\y);
}
\draw[thick] (n2) -- (n) -- (n1) -- (n11) (n1) -- (n12) -- (n121) (n12) -- (n122);
\foreach \c in {,1,11,12,121,122,2}{
 \draw[fill=white,thick] (n\c) circle (0.45);
}
\foreach \c in {1,11,12,121,122,2}{
 \node at (n\c) {$n_{\c}$};
}
\node at (n) {$n_{\vphantom{1}}$};

\begin{scope}[shift={(5.5,-0.5)}]
\foreach \x/\y/\c in {0/0/11,1/1/1,2/0/121,3.5/0/122,4.5/1/,5.5/0/2,2.75/2/12}{
 \coordinate (n\c) at (0.75*\x,\y);
}
\draw[thick] (n2) -- (n) -- (n122) (n) -- (n12) -- (n1) (n121) -- (n1) -- (n11);
\foreach \fil/\list in {gray!30/{,1,12},white/{11,2,121,122}}{
 \foreach \c in \list{
  \draw[fill=\fil,thick] (n\c) circle (0.45);
}}
\foreach \c in {1,2,11,12,121,122}{
 \node at (n\c) {$n_{\c}$};
}
\node at (n) {$n_{\vphantom{1}}$};
\end{scope}

\begin{scope}[shift={(-6.25,0)}]
\foreach \x/\y/\c in {1/1/11,2/2/1,3/1/,2/0/12,4/0/2,1/-1/121,3/-1/122}{
 \coordinate (n\c) at (0.75*\x,\y);
}
\draw[thick] (n11) -- (n1) -- (n) (n12) -- (n) -- (n2) (n121) -- (n12) -- (n122);
\foreach \fil/\list in {gray!30/{,1},white/{11,12,2,121,122}}{
 \foreach \c in \list{
  \draw[fill=\fil,thick] (n\c) circle (0.45);
}}
\foreach \c in {1,2,11,12,121,122}{
 \node at (n\c) {$n_{\c}$};
}
\node at (n) {$n_{\vphantom{1}}$};
\end{scope}

\draw[line width=3pt,->,>=stealth] (3,0.5) --++ (1.75,0);
\draw[line width=3pt,->,>=stealth] (-0.75,0.5) --++ (-1.75,0);

\node[anchor=south] at (-1.475,0.5) {simple};
\node[anchor=north] at (-1.475,0.45) {rotation};
\node[anchor=south] at (3.725,0.5) {double\vphantom{p}};
\node[anchor=north] at (3.725,0.45) {rotation};
\end{tikzpicture}
\end{center}
\caption{Performing a simple or double rotation.
Affected nodes are coloured in grey.
\label{fig:balancing}}
\end{figure}

\begin{routine}[$\sA$-balancing]\label{routine:A}
Given real numbers~$u$,~$v$,~$w$ and~$x$, the~$\sA(u,v,w,x)$-balancing algorithm operates as follows on the binary tree~$\T$ it receives as input:
\begin{enumerate}[itemsep=0.5pt,topsep=0.5pt]
\item if~$|n_1| > u |n| + w$ and~$|n_{11}| \geqslant (1-v) |n| - x$, perform the simple rotation of Figure~\ref{fig:balancing};
\item if~$|n_2| > u |n| + w$ and~$|n_{22}| \geqslant (1-v) |n| - x$, perform the mirror image of that rotation;
\item if~$|n_1| > u |n| + w$ and~$|n_{11}| < (1-v) |n| - x$, perform the double rotation of Figure~\ref{fig:balancing};
\item if~$|n_2| > u |n| + w$ and~$|n_{22}| < (1-v) |n| - x$, perform the mirror image of that rotation;
\item if~$\max\{|n_1|,|n_2|\} \leqslant u |n| + w$, do not modify~$\T$.
\end{enumerate}
\end{routine}

\begin{routine}[$\sB$-balancing]\label{routine:B}
Given real numbers~$y$ and~$z$, the~$\sB(y,z)$-balancing algorithm operates as follows on the binary tree~$\T$ it receives as input:
\begin{enumerate}[itemsep=0.5pt,topsep=0.5pt]
\item if~$|n_{11}| > y |n| + z$, perform the simple rotation shown in Figure~\ref{fig:balancing};
\item if~$|n_{22}| > y |n| + z$, perform the mirror image of that rotation;
\item if~$|n_{12}| > y |n| + z$, perform the double rotation shown in Figure~\ref{fig:balancing};
\item if~$|n_{21}| > y |n| + z$, perform the mirror image of that rotation;
\item if~$\max\{|n_{11}|,|n_{12}|,|n_{21}|,|n_{22}|\} \leqslant y |n| + z$, do not modify~$\T$.
\end{enumerate}
\end{routine}

Parameters~$w$,~$x$ and~$z$ can be seen as terms governing some error margin, that we will set equal to zero in Section~\ref{sec:BU} and will be non-zero in Section~\ref{sec:TD}.
Indeed, in the former case, we develop bottom-up updating algorithms, and thus have a perfect knowledge of~$\T$.
By contrast, in the latter case, we develop top-down algorithms and thus cannot yet know where the leaf~$s$ should be inserted or deleted: depending on the answer, the tree~$\T$ may have different shapes, and we cannot anticipate which will be chosen.

The idea of the~$\sA$-balancing algorithm is to either check that our tree root is~$\cbal{u}$ or to transform it into a~$\cbal{v}$ node; in the latter case, each affected node will also be made~$\cbal{v}$, which will improve the \cbal{} of our tree root by at least~$u-v$ without damaging the balances of other nodes too much.
If we perform a simple rotation,~$n_{11}$ will be a new root child, and it will definitely not be too large, but its new sibling, of weight~$|n|-|n_{11}|$, might prevent our new root from being~$\cbal{v}$;
this undesirable case should occur when~$|n_{11}| < (1-v) |n|$, which is why, in that case, we perform a double rotation instead of a simple one.
More precisely, here is a result that can be stated about the~$\sA$-balancing algorithm.

\begin{lemma}\label{lem:A-balancing}
Let~$\alpha$ be a real number such that~$1 / \sqrt{2} \leqslant \alpha < 3/4$.
Then, let~$\alphab = 19/24$ and
\vspace{-1mm}
\[\alphaa = \frac{1+\alpha-\sqrt{(1-\alpha)(5-\alpha)}}{2(2\alpha-1)}.\]

\vspace{-2mm}%
We have~$1/\sqrt{2} \leqslant \alphaa \leqslant \alpha$.
Moreover, when~$\alphaa \leqslant u \leqslant \alpha$, the five cases in which the~$\sA(u,\alphaa,0,0)$-balancing algorithm consists are pairwise incompatible, and those rotations they trigger can always be performed;
the algorithm itself transforms each~$\GCBB{\alphab}{\alpha}$-tree~$\T$ whose root is not~$\cbal{u}$ into a~$\GCB{\alpha}$-tree~$\T'$ such that~$\balc(\T') \leqslant \balc(\T)$ and whose affected nodes are~$\cbal{\alphaa}$.
\end{lemma}

\begin{proof}
Below, we will use the following (in)equalities, which are all easy to check with any computer algebra system (whereas some are quite tedious to check by hand) whenever~$1/\sqrt{2} \leqslant \alpha < 3/4$.
Except the first two inequalities, each of them is labelled and later reused to prove another inequality with the same label; equality~(2.1\,+\,2.5) is used to prove both subsequent inequalities~(2.1) and~(2.5).
\newcommand{\vjline}{\vphantom{(\alpha^2\alphaa u)}}
\vspace{-6mm}

\begin{multicols}{2}
\begin{align}
& \vjline \text{\rlap{$1/\sqrt{2} \leqslant \alphaa;$}}\hphantom{(1-\alphaa)(\alphab-1+\alphaa) < \alphaa(1-\alphab);} \notag \\
& \vjline \alpha \alphab < \alphaa; \tag{1.1} \\
& \vjline (1-\alphaa)^2 < \alphaa^2 (1-\alpha); \tag{1.3} \\
& \vjline (1-\alphaa)\alpha^2 < \alphaa(1-\alpha); \tag{2.2} \\
& \vjline (1-\alphaa)\alpha^2\alphab < \alphaa(1-\alphab); \tag{2.4}
\end{align}%

\begin{align}
& \vjline \alphaa \leqslant \alpha; \notag \\
& \vjline (1-\alphaa)(\alphab-1+\alphaa) < \alphaa(1-\alphab); \tag{1.2} \\
& \vjline (1-\alphaa)^2 = \alphaa (1-\alpha)(2\alphaa - 1); \tag{2.1\,+\,2.5} \\
& \vjline (1-\alphaa)(1-\alpha) + \alpha \alphab < \alphaa; \tag{2.3} \\
& \vjline 1 + (\alpha^2-1)\alphaa < \alphaa. \tag{2.6} \\[-30mm] \notag
\end{align}
\end{multicols}

This already proves the inequality~$1/\sqrt{2} \leqslant \alphaa \leqslant \alpha$ of Lemma~\ref{lem:A-balancing}.

Then, let~$|m|$ denote the weight of a node~$m$ in the tree~$\T$, and let~$|m|'$ denote its weight in~$\T'$.
When these weights are equal, we will prefer using the notation~$|m|$ even when considering the weight of~$m$ in~$\T'$.
The root of~$\T$ is denoted by~$n$, its left and right children are denoted by~$n_1$ and~$n_2$, and so on.
Hence, when~$u \geqslant \alphaa$, we have~$2 u |n| \geqslant \sqrt{2} |n| \geqslant |n_1| + |n_2|$, which makes the inequalities~$|n_1| > u |n|$ and~$|n_2| > u |n|$ incompatible.

Proving that the desired rotations can indeed be performed amounts to showing that~$n$ and~$n_1$ are actual tree nodes (instead of spurious empty nodes of weight~$1/2$ or less) in case~1, and that~$n_{12}$ is also a tree node in case~3; cases~2 and~4 will be treated symmetrically.
In cases~1 and~3, since~$|n_1| > u |n| > |n| / 2$, the node~$n$ is an actual tree node, and cannot be a leaf, which means, as desired, that~$n$ and~$n_1$ are tree nodes;
furthermore,~$|n_1| > u (|n_1|+1) > (|n_1|+1) / \sqrt{2}$, i.e.,~$|n_1| > 1 + \sqrt{2} > 2$, and therefore~$|n_1| \geqslant 3$ and~$|n| \geqslant 4$.
Then, in case~3, we also have~$|n_{12}| = |n_1| - |n_{11}| > u |n| - (1-\alphaa) |n| \geqslant (\sqrt{2}-1) |n| > 1$, which means that~$n_{12}$ is also a tree node.

It remains to prove that each affected node is~$\cbal{\alphaa}$: if the root of~$\T$ was not~$\cbal{u}$, we will have~$\balc(\T') \leqslant \alphaa \leqslant u \leqslant \balc(\T)$.

When~$\alphab |n| \geqslant |n_1| > u |n| \geqslant \alphaa |n|$ and~$|n_{11}| \geqslant (1-\alphaa) |n|$, a simple rotation is performed, and
{\allowdisplaybreaks
\begin{align}
|n_{11}| & \leqslant \alpha |n_1| \leqslant \alpha \alphab |n| \leqslant \alphaa |n| = \alphaa |n_1|'; \tag{1.1} \\
(1-\alphaa) |n_{12}| & = (1-\alphaa)(|n_1| - |n_{11}|) \leqslant (1-\alphaa)(\alphab |n| - (1-\alphaa)|n|) \notag \\
& \leqslant \alphaa(1-\alphab)|n| \leqslant \alphaa(|n|-|n_1|) = \alphaa |n_2|; \tag{1.2} \\
(1-\alphaa)|n_2| & = (1-\alphaa)(|n|-|n_1|) \leqslant (1-\alphaa)(1-\alphaa) |n| \notag \\
& \leqslant \alphaa(1-\alpha)\alphaa |n| \leqslant \alphaa(1-\alpha) |n_1| \leqslant \alphaa(|n_1|-|n_{11}|) = \alphaa |n_{12}|; \tag{1.3} \\
|n|' & = |n| - |n_{11}| \leqslant |n| - (1-\alphaa)|n| = \alphaa |n| =  \alphaa |n_1|', \tag{1.4}
\end{align}
}
which means precisely that~$n$ and~$n_1$ are~$\cbal{\alphaa}$ in~$\T'$.

Similarly, when~$\alphab |n| \geqslant |n_1| > u |n| \geqslant \alphaa |n|$ and~$|n_{11}| < (1-\alphaa) |n|$, a double rotation is performed, and
\begin{align}
(1-\alphaa)|n_{11}| & \leqslant (1-\alphaa)^2 |n| = \alphaa (1-\alpha) (\alphaa |n| - (1-\alphaa)|n|) \notag \\
& \leqslant \alphaa (1-\alpha)(|n_1| - |n_{11}|) = \alphaa (1-\alpha) |n_{12}| \notag \\
& \leqslant \alphaa(|n_{12}|-|n_{122}|) = \alphaa |n_{121}|; \tag{2.1} \\
(1-\alphaa) |n_{121}| & \leqslant (1-\alphaa)\alpha |n_{12}| \leqslant (1-\alphaa) \alpha^2 |n_1| \notag \\
& \leqslant \alphaa(1-\alpha)|n_1| \leqslant \alphaa(|n_1|-|n_{12}|) = \alphaa |n_{11}|; \tag{2.2} \\
|n_1|' & = |n_{11}| + |n_{121}| \leqslant |n_{11}| + \alpha |n_{12}| = (1-\alpha)|n_{11}| + \alpha |n_1| \notag \\
& \leqslant (1-\alpha)(1-\alphaa)|n| + \alpha \alphab|n| \leqslant \alphaa |n| = \alphaa |n_{12}|'; \tag{2.3} \\
(1-\alphaa)|n_{122}| & \leqslant (1-\alphaa)\alpha |n_{12}| \leqslant (1-\alphaa)\alpha^2 |n_1| \leqslant (1-\alphaa)\alpha^2 \alphab |n| \notag \\
& \leqslant \alphaa (1-\alphab)|n| = \alphaa |n_2|; \tag{2.4} \\
(1-\alphaa) |n_2| & = (1-\alphaa)(|n|-|n_1|) \leqslant (1-\alphaa)^2|n| = \alphaa(1-\alpha)(\alphaa|n| - (1-\alphaa)|n|) \notag \\
& \leqslant \alphaa(1-\alpha)(|n_1|-|n_{11}|) = \alphaa(1-\alpha)|n_{12}| \notag \\
& \leqslant \alphaa(|n_{12}|-|n_{121}|) = \alphaa |n_{122}|; \tag{2.5} \\
|n|' & = |n| - |n_1| + |n_{122}| \leqslant |n| - |n_1| + \alpha |n_{12}| \notag \\
& \leqslant |n| + (\alpha^2 - 1) |n_1| \leqslant |n| + (\alpha^2-1) \alphaa |n| \leqslant \alphaa |n| = \alphaa |n_{12}|', \tag{2.6}
\end{align}
which means precisely that~$n$,~$n_1$ and~$n_{12}$ are~$\cbal{\alphaa}$ in~$\T'$.

Finally, the cases where~$|n_2| > u |n|$ are symmetrical, and therefore the conclusion of Lemma~\ref{lem:A-balancing} is valid in these cases too.
\end{proof}

Similarly, the~$\sB$-balancing algorithm aims at making our tree root~$\gcbal{y}$.
If our tree is too unbalanced, we will either promote the responsible grandchild to being a child (if this grandchild was~$n_{11}$ or~$n_{22}$) with a simple rotation, or split it in two (if this grandchild was~$n_{12}$ or~$n_{21}$) with a double rotation.
For adequate values of the parameters~$u$,~$v$ and~$y$, each node affected by the~$\sB$-balancing algorithm will be~$(v,\hat{y})$-balanced for some real number~$\hat{y} \leqslant y$ (that does not need to be given as parameter), thus avoiding any damage to the \gcbal{} of affected nodes, and improving that of our tree root by at least~$y-\hat{y}$.
Hence, here is a result that can be stated about the~$\sB$-balancing algorithm;
its proof is similar to that of Lemma~\ref{lem:A-balancing}, and can be found in appendix.

\begin{restatable}{lemma}{bbalancing}\label{lem:B-balancing}
Let~$\alpha$ and~$\beta$ be real numbers such that~$1/\sqrt{2} \leqslant \alpha < 3/4$ and~$\scrB(\alpha) \leqslant \beta \leqslant \alpha^2$.
Then, let
\[\alphaa = \frac{1-2\alpha+6\alpha^2-\sqrt{(1-\alpha)(5-\alpha)}}{5(2\alpha-1)} \text{, } \betaa = \frac{1+\alpha - \sqrt{(1-\alpha)^2+4\alpha(\alpha^2-\beta)}}{2(1-\alpha^2+\beta)}\alpha,\]
as well as~$\delta = \min\{\alphaa,\betaa/\alpha\}$.

We have~$1/4 < \betaa \leqslant \beta$.
Moreover, when~$\alphaa \leqslant u \leqslant \alpha$ and~$\betaa \leqslant y \leqslant \beta$, if we apply the~$\sB(y,0)$-balancing algorithm on a~$\GCBB{\alphaa,1}{\alpha,\beta}$-tree~$\T$, cases~1 to~5 are pairwise incompatible, and those rotations they trigger can always be performed; if the root of~$\T$ is not~$\gcbal{y}$, the algorithm transforms~$\T$ into a~$\GCBB{u,y}{\alpha,\beta}$-tree~$\T'$ such that~$\balc(\T') \leqslant \balc(\T)$ and whose affected nodes are~$(\delta,\betaa)$-balanced.
\end{restatable}

Finally, these building blocks can be combined to ensure that a $\GCB{\alpha,\beta}$-tree whose root is suddenly slightly out of line will be rebalanced.

\begin{routine}[$\sA\sB$-balancing]\label{alg:AB}
Given real numbers~$u$,~$v$,~$y$ and~$\hat{y}$, the~$\sA\sB(u,v,y,\hat{y})$-balancing algorithm executes the following operations on the tree~$\T$ it receives as input:
\begin{enumerate}[itemsep=0.5pt,topsep=0.5pt]
\item the~$\sA(u,v,0,0)$-balancing algorithm transforms~$\T$ into a tree~$\T^{(1)}$;
\item if~$\T \neq \T^{(1)}$, 
\begin{itemize}[itemsep=0.5pt,topsep=0.5pt]
\item the~$\sB(\hat{y},0)$-balancing algorithm is applied (in parallel) to the left and right sub-trees of~$\T^{(1)}$, which transforms~$\T^{(1)}$ into a tree~$\T^{(2)}$;
\item then, the~$\sB(\hat{y},0)$-balancing algorithm is applied to~$\T^{(2)}$ itself;
\end{itemize}
\item otherwise, the~$\sB(y,0)$-balancing algorithm is applied to~$\T$.
\end{enumerate}
\end{routine}

\begin{proposition}\label{pro:AB-balancing}
Let~$\alpha$ and~$\beta$ be real numbers such that~$1/\sqrt{2} \leqslant \alpha < 3/4$ and~$\scrB(\alpha) \leqslant \beta \leqslant \alpha^2$; we recall that~$\scrB(\alpha) = (\sqrt{1+4\alpha}-1)/2$.
Then, let~$\alphab = 19/24$,
\[\alphaa = \frac{1-2\alpha+6\alpha^2-\sqrt{(1-\alpha)(5-\alpha)}}{5(2\alpha-1)} \text{ and } \betaa = \frac{1+\alpha - \sqrt{(1-\alpha)^2+4\alpha(\alpha^2-\beta)}}{2(1-\alpha^2+\beta)}\alpha.\]

When~$\alphaa \leqslant u \leqslant \alpha$ and~$\betaa \leqslant y \leqslant \beta$, the~$\sA\sB(u,\alphaa,y,\betaa)$-balancing algorithm, when applied to a~$\GCBB{\alphab,1}{\alpha,\beta}$-tree~$\T$ whose root is not~$(u,y)$-balanced, transforms~$\T$ into a~$\GCBB{\alphaa,\betaa}{\alpha,\beta}$-tree~$\T'$ such that~$\balc(\T') \leqslant \balc(\T)$ and whose affected nodes are all~$(\alphaa,\betaa)$-balanced.
\end{proposition}

\begin{proof}
If the root of~$\T$ is~$(u,y)$-balanced, the~$\sA\sB$-balancing algorithm does nothing.
If this root is~$\cbal{u}$ but not~$\gcbal{y}$, the~$\sB(y,0)$-balancing algorithm is called; Lemma~\ref{lem:B-balancing} ensures that every affected node will be~$(\alphaa,\betaa)$-balanced and that~$\balc(\T') \leqslant \balc(\T)$.

Finally, if the root of~$\T$ it is not~$\cbal{u}$, we will call the~$\sA(u,\alphaa,0,0)$-algorithm, obtaining a~$\GCB{\alpha,\beta}$-tree~$\T^{(1)}$, whose root will be denoted by~$r$.
Nodes affected by this call are~$r$ and one or two of its children~$r_1$ and~$r_2$.
The~$\sB(\betaa,0)$-balancing algorithm is then applied to the~$\GCBB{\alpha,1}{\alpha,\beta}$-trees rooted at~$r_1$ and~$r_2$;
as a result,~$\T^{(2)}$ is a~$\GCBB{\alpha,1}{\alpha,\beta}$-tree rooted at~$r$, to which the~$\sB(\betaa,0)$-balancing algorithm is applied once more.
No tree to which the~$\sA$- and~$\sB$-balancing algorithms were applied saw its~$\cbal{}$ increase, and thus~$\balc(\T') \leqslant \balc(\T)$. 

Moreover, each node~$m$ affected by the~$\sA\sB$-balancing algorithm is either affected by some call to the~$\sB$-balancing algorithm or the root of some tree that the~$\sB$-balancing algorithm did not modify; the latter case may only concern nodes~$r$,~$r_1$ and~$r_2$.
In both cases,~$m$ ends up being~$(\alphaa,\betaa)$-balanced, which completes the proof.
\end{proof}

Now is the time when we actually describe our update algorithm, which we will apply to a~$\GCB{\alpha,\beta}$-tree in which a leaf~$s$ is to be inserted, or from which~$s$ is to be deleted.
We have two regimes: the first one concerns trees with~$11$ nodes or less, which we can just restructure in order to make them as balanced as possible, and the second regime concerns trees with~$12$ nodes or more.

\begin{routine}[Bottom-up update]\label{alg:BU-rebalancing}
Let~$\alpha$,~$\alphaa$ and~$\beta$ be real numbers given in Proposition~\ref{pro:AB-balancing}.
Our updating algorithm executes the following operations on the tree~$\T$ it receives as input:
\begin{itemize}[itemsep=0.5pt,topsep=0.5pt]
\item if~$\T$ contains~$11$ nodes or less, insert or delete~$s$ (if possible) and make the resulting tree a perfectly balanced tree, i.e., a tree in which the weights of two siblings differ by at most~$1$;
\item if~$\T$ contains~$12$ nodes or more,
\begin{enumerate}[itemsep=0.5pt,topsep=0.5pt]
\item recursively rebalance the (left or right) sub-tree in which~$s$ lies, and then
\item if the update turned out to be non-redundant, update the weight of the root of~$\T$, and then apply the~$\sA\sB(\alpha,\alphaa,\beta,\betaa)$-balancing algorithm to~$\T$.
\end{enumerate}
\end{itemize}
\end{routine}

Due to the recursive flavour of this algorithm, the rebalancing operations are performed bottom-up.
We will see in Section~\ref{sec:TD} how to perform these operations top-down.

\begin{theorem}\label{thm:BU-balancing}
Let~$\alpha$ and~$\beta$ be real numbers such that~$1/\sqrt{2} \leqslant \alpha < 3/4$ and~$\scrB(\alpha) \leqslant \beta \leqslant \alpha^2$.
The tree obtained by using Algorithm~\ref{alg:BU-rebalancing} to insert a leaf in a~$\GCB{\alpha,\beta}$-tree~$\T$ or delete a node from~$\T$ is also a~$\GCB{\alpha,\beta}$-tree.
\end{theorem}

\begin{proof}
Our proof is a variant of the proof of~\cite{BM80}.

Let~$\T'$ and~$\T''$ be the trees obtained just after step~1 and step~2 of Algorithm~\ref{alg:BU-rebalancing}, respectively.
We will prove by induction on~$|\T|$ that~$\T''$ is a~$\GCB{\alpha,\beta}$-tree.
The result being correct by construction when~$|\T| \leqslant 12$, we assume that~$|\T| \geqslant 13$.

Let~$\T_1$ and~$\T_2$ be the left and right sub-trees of~$\T$, and let~$t = |\T|$,~$t_1 = |\T_1|$ and~$t_2 = |\T_2|$.
Similarly, let~$\T'_1$ and~$\T'_2$ be the left and right sub-trees of~$\T'$, and let~$t' = |\T'|$~$t'_1 = |\T'_1|$ and~$t'_2 = |\T'_2|$.
Without loss of generality, we assume that the update was non-redundant, and that the tree~$\T$ was altered by adding~$s$ to~$\T_1$ or by deleting~$s$ from~$\T_2$.
This means that either~$(t',t'_1,t'_2) = (t+1,t_1+1,t_2)$ or~$(t'_1,t'_2) = (t-1,t_1,t_2-1)$.
In both cases, the induction hypothesis ensures that~$\T'_1$ and~$\T'_2$ are~$\GCB{\alpha,\beta}$-trees.

We prove now that the root of~$\T'$ is~$\cbal{\alphab}$, thereby allowing us to use Proposition~\ref{pro:AB-balancing} and completing the induction.
Indeed, let~$t = t_1 + t_2$ and~$t' = t'_1 + t'_2$.
Since~$4 t_1 \leqslant 4 \alpha t < 3 t$, we know that~$4 t_1 \leqslant 3 t-1$.

Thus, in case of an insertion,~$t'_1 = t_1+1 \leqslant (3 t+3)/4 = 3 t' / 4 \leqslant \alphab t'$ and~$t'_2 = t_2 \leqslant \alpha t \leqslant \alphab t'$.
Similarly, in case of a deletion,~$t'_1 = t_1 \leqslant (3t-1)/4 = (3t'+2)/4 = \alphab t' - (t'-12)/24 \leqslant \alphab t'$, whereas~$t'_2 = t_2 - 1 \leqslant \alpha t - 1 \leqslant \alpha t' \leqslant \alphab t'$.
\end{proof}

\section{Top-down algorithm}\label{sec:TD}

In this section, we propose a top-down algorithm for inserting an element into (or deleting an element from) a weight-balanced tree.
This algorithm is valid whenever~$1/\sqrt{2} \leqslant \alpha < 3/4$ and~$\scrB(\alpha) \leqslant \beta \leqslant \alpha^2$, thereby completing the algorithm of Lai and Wood~\cite{LW93}, which works only when~$3/4 \leqslant \alpha \leqslant 9/11$ and~$\beta = \alpha^2$.
This new algorithm is inspired by theirs: we wish to perform a top-down restructuring pass while adjusting weight information.
If the update is redundant, a second top-down pass will be needed to update this information, but no further restructuring will be needed.

More precisely, we aim at having a top-down algorithm that requires considering only a constant number of tree nodes at each step; this number may not depend on the parameters such as~$\alpha$ and~$\beta$.
Moreover, we still wish, by using this algorithm, to perform only a constant number of rotations per update; this number \emph{may} depend on~$\alpha$,~$\beta$ and other parameters.
Consequently, the idea of the algorithm, which will be made more precise in Algorithm~\ref{alg:TD-updating}, is as follows:
\begin{enumerate}[itemsep=0.5pt,topsep=0.5pt]
\item If the tree is large enough, we may rebalance it to make sure it will remain balanced even if we recursively update one of its children.
Like in Section~\ref{sec:BU}, rotations performed in this phase should be so efficient that only an amortised constant number of rotations per update will be useful.

\item If the tree is small enough, we make it as balanced as possible, the notion of being top-down being void in this case.
\end{enumerate}

A key object towards defining and proving the correctness of our algorithm is the notion of \emph{robust} grandchildren-balanced trees.

\begin{definition}\label{def:GBR-trees}
Let~$n$ be a node of a binary tree~$\T$, let~$n_1$ and~$n_2$ be its children, and let~$n_{11}$,~$n_{12}$,~$n_{21}$ and~$n_{22}$ be its grandchildren.
Given real numbers~$\alpha$ and~$\beta$, we say that~$n$ is \emph{robustly}~$\cbal{\alpha}$ when
\[\max\{|n_1|,|n_2|\}+1 \leqslant \alpha(|n|+1) \text{ and } \max\{|n_1|,|n_2|\} \leqslant \alpha(|n|-1);\]
that~$n$ is \emph{robustly}~$\gcbal{\beta}$ when
\[\max\{|n_{11}|,|n_{12}|,|n_{21}|,|n_{22}|\}+1 \leqslant \beta(|n|+1) \text{ and } \max\{|n_{11}|,|n_{12}|,|n_{21}|,|n_{22}|\} \leqslant \beta(|n|-1);\]
and that~$n$ is \emph{robustly}~$(\alpha,\beta)$-balanced when~$n$ is robustly~$\cbal{\alpha}$ and robustly~$\gcbal{\beta}$.

Finally, we say that a tree~$\T$ is a~$\RGCBB{x,y}{\alpha,\beta}$-tree when its root is robustly~$(x,y)$-balanced and its other nodes are~$(\alpha,\beta)$-balanced.
For the sake of concision, we simply say that~$\T$ is a~$\RGCB{\alpha,\beta}$-tree or a~$\RGCB{\alpha}$-tree when~$(x,y) = (\alpha,\beta)$ or~$(x,y,\beta) = (\alpha,1,1)$, respectively.
\end{definition}

Inequalities~$|n_i| + 1 \leqslant \alpha(|n|+1)$ and~$|n_i| \leqslant \alpha(|n|-1)$ ensure that, even if a leaf is inserted into or deleted from~$\T$, the node~$n$ will remain~$\cbal{\alpha}$.
Inequalities~$|n_{ij}| + 1 \leqslant \beta(|n|+1)$ and~$|n_{ij}| \leqslant \beta(|n|-1)$ serve the same purpose, but for ensuring that~$n$ remains~$\gcbal{\beta}$.
Finally, each tree node is robustly~$(1,1)$-balanced.

Setting~$\sR(t) = \min\{-t,t-1\}$ and observing that~$\min\{t(|n|+1)-1,t(|n|-1)\} = t |n| + \sR(t)$ for all real numbers~$t$ provides us with a more succinct characterisation of~$(\alpha,\beta)$-balanced nodes, which will require giving non-zero values to the parameters~$w$,~$x$ and~$z$:
the node~$n$ is robustly~$(\alpha,\beta)$-balanced if and only if~$\max\{|n_1|,|n_2|\} \leqslant \alpha |n| + R(\alpha)$ and~$\max\{|n_{11}|,|n_{12}|,|n_{21}|,|n_{22}|\} \leqslant \beta|n| + \sR(\beta)$.

We can now group several calls to the~$\sA$- and~$\sB$-balancing algorithms into a so-called~$\sR\sA\sB$-balancing algorithm, whose aim is to make a tree root robustly balanced.

\begin{routine}[$\sR\sA\sB$-balancing]\label{alg:RAB}
Given real numbers~$u$,~$v$,~$y$ and~$\hat{y}$, the~$\sR\sA\sB(u,v,y,\hat{y})$-balancing algorithm executes the following operations, in this order, on the tree~$\T$ it receives as input, thus transforming it into a tree~$\T'$:
\begin{enumerate}[itemsep=0.5pt,topsep=0.5pt]
\item the~$\sA(u,v,\sR(u),\sR(v))$-balancing algorithm transforms~$\T$ into a tree~$\T^{(1)}$;
\item if~$\T \neq \T^{(1)}$,
\begin{itemize}[itemsep=0.5pt,topsep=0.5pt]
\item the~$\sA\sB(v,v,\hat{y},\hat{y})$-balancing algorithm is applied (in parallel) to the left and right sub-trees of~$\T^{(1)}$, which transforms~$\T^{(1)}$ into a tree~$\T^{(2)}$;
\item the~$\sB(\hat{y},\sR(\hat{y}))$-balancing algorithm transforms~$\T^{(2)}$ into a tree~$\T^{(3)}$;
\item the~$\sA\sB(v,v,\hat{y},\hat{y})$-balancing algorithm is applied (in parallel) to the left and right sub-trees of~$\T^{(3)}$, which yields the tree~$\T'$;
\end{itemize}
\item otherwise, the~$\sB(y,\sR(y))$-balancing algorithm transforms~$\T$ into a tree~$\T^{(3)}$, and then,
\begin{itemize}[itemsep=0.5pt,topsep=0.5pt]
\item if~$\T \neq \T^{(3)}$, the~$\sA\sB(v,v,\hat{y},\hat{y})$-balancing algorithm is applied (in parallel) to the left and right sub-trees of~$\T^{(3)}$, which yields the tree~$\T'$;
\item otherwise,~$\T' = \T$.
\end{itemize}
\end{enumerate}
\end{routine}

The idea is similar to that of the plain~$\sA$- and~$\sB$-balancing algorithm, but we wish to transform a tree whose root is not robustly~$(u,y)$-balanced into a tree~$\T'$ whose root will be robustly~$(v,\hat{y})$-balanced.
Calling the~$\sA(u,v,\sR(u),\sR(v))$-,~$\sB(\hat{y},\sR(\hat{y}))$- and~$\sB(y,\sR(y))$-balancing algorithms achieves this goal, while possibly damaging the balance of the children of the root of~$\T'$;
we solve this issue by the~$\sA\sB$-balancing algorithm to the sub-trees rooted at these children.

\begin{restatable}{proposition}{rabbalancing}\label{pro:RAB-balancing}
Let~$\alpha$,~$\alphaa$,~$\beta$ and~$\betaa$ be real numbers given in Proposition~\ref{pro:AB-balancing}.
The~$\sR\sA\sB(\alpha,\alphaa,\beta,\betaa)$-balancing algorithm transforms each~$\GCBB{\alpha,1}{\alpha,\beta}$-tree~$\T$ with~$30$ nodes or more and whose root is not robustly~$(\alpha,\beta)$-balanced into a~$\RGCB{\alpha,\beta}$-tree~$\T'$ such that~$\balc(\T') \leqslant \balc(\T)$ and whose affected nodes are all~$(\alphaa,\betaa)$-balanced.
\end{restatable}

We finally describe our top-down updating algorithm, which we will apply to a~$\GCB{\alpha,\beta}$-tree in which a leaf~$s$ is to be inserted, or from which~$s$ is to be deleted.
Like the algorithm of Lai and Wood~\cite{LW93}, this algorithm is not purely top-down, because (i) if we wanted to delete an internal node~$a$, as illustrated in Figure~\ref{fig:in-out},
it requires maintaining a pointer to~$a$, whose key will later be replaced by another node key, and~(ii) if the update turned out to be redundant, a second top-down pass will be required to cancel every weight update we performed.
Alternative representations of weight-balanced trees or~$\GCB{\alpha,\beta}$-trees allow omitting this second pass:
instead of storing the weight~$|n|$ at each node~$n$, we should store one of the weights~$|n_1|$ or~$|n_2|$ as well as a bit indicating which weight we stored.

\begin{routine}[Top-down update]\label{alg:TD-updating}
Let~$\alpha$,~$\alphaa$,~$\beta$ and~$\betaa$ be real numbers given in Proposition~\ref{pro:AB-balancing}.
Our updating algorithm executes the following operations on the tree~$\T$ it receives as input:
\begin{itemize}[itemsep=0.5pt,topsep=0.5pt]
\item if~$\T$ contains~$29$ nodes or less, insert or delete~$s$ (if possible) and make the resulting tree a perfectly balanced tree, i.e., a tree in which the weights of two siblings differ by at most~$1$;
\item if~$\T$ contains~$30$ nodes or more,
\begin{enumerate}[itemsep=0.5pt,topsep=0.5pt]
\item apply the~$\sR\sA\sB(\alpha,\alphaa,\beta,\betaa)$-balancing algorithm to~$\T$, then
\item update the weight of the root of~$\T$, and finally
\item recursively update the (left or right) sub-tree of~$\T$ that contains~$s$.
\end{enumerate}
\end{itemize}
\end{routine}

\begin{restatable}{theorem}{tdbalancing}\label{thm:TD-balancing}
Let~$\alpha$ and~$\beta$ be real numbers such that~$1/\sqrt{2} \leqslant \alpha < 3/4$ and~$\scrB(\alpha) \leqslant \beta \leqslant \alpha^2$.
The tree obtained by using Algorithm~\ref{alg:TD-updating} to insert a leaf in a~$\GCB{\alpha,\beta}$-tree~$\T$ or delete a node from~$\T$ is also a~$\GCB{\alpha,\beta}$-tree.
\end{restatable}

\section{Complexity analysis}\label{sec:complexity}

In this final section, we prove that Algorithms~\ref{alg:BU-rebalancing} and~\ref{alg:TD-updating} proposed above perform an amortised constant number of rotations per attempted insertion or deletion.

The idea consists in evaluating the sums of child- and grandchild-balances of all tree nodes.
Indeed, inserting or deleting a leaf~$s$ will slightly damage the child- and grandchild balances of the ancestors of~$s$, which may increase the sum of these balances by no more than a constant.
In the opposite direction, whenever the~$\sA\sB$- or~$\sR\sA\sB$-balancing algorithm changes a sub-tree~$\T$, some node will see its child-balance decrease by approximately~$\alpha - \alphaa$, or its grandchild-balance decrease by approximately~$\beta - \betaa$;
other node balances might be damaged in the process, but not to the point of exceeding~$(\alphaa,\betaa)$.
Thus, not too many changes may be performed.

These ideas lead to the following result.

\begin{restatable}{theorem}{bucomplexity}\label{thm:BU-complexity}
Let~$\alpha$ and~$\beta$ be real numbers such that~$1 / \sqrt{2} < \alpha < 3/4$ and~$\scrB(\alpha) < \beta \leqslant \alpha^2$.
Let~$\eta$ and~$\varepsilon$ be given by~$\eta = \alpha - 1/\sqrt{2}$ and~$\varepsilon = \min\{\beta - \scrB(\alpha),\alpha^2-\beta\}$ if~$\scrB(\alpha) < \beta < \alpha^2$, or~$\varepsilon = 1$ if~$\beta = \alpha^2$.

Then, let~$\T_0, \T_1, \ldots, \T_k$ be~$\GCB{\alpha,\beta}$-trees defined as follows:~$\T_0$ is the empty tree and, for all~$\ell \geqslant 0$, the tree~$\T_{\ell+1}$ is obtained by inserting a leaf in~$\T_\ell$ or deleting a leaf from~$\T_\ell$, using either Algorithm~\ref{alg:BU-rebalancing} or~\ref{alg:TD-updating} to do so.
Gradually transforming~$\T_0$ into~$\T_k$ via such steps requires only~$\O(k/\eta +k/\varepsilon)$ rotations.
\end{restatable}

\begin{proof}[Proof outline]
Each node~$n$ is given a real-valued counter~$\scc(n)$ that receives the value $0$ when~$n$ is created (i.e., inserted in a tree) or affected by a rotation, and is increased by $2/|n|$ when a descendant of $n$ is about to be created or deleted.

Finally, let~$\sC$ be the sum of all these counters $\scc(n)$, and let~$\delta = \min\{\alpha - \alphaa,\beta - \betaa,1/62\}/2$.
One can prove that the sum~$\sC$ increases by no more than $16$ when a leaf is inserted or deleted, and decreases by at least $\delta$ when the $\sA\sB$- or $\sR\sA\sB$-balancing algorithms trigger a rotation and the tree contains at least $32$ nodes;
if the tree contains fewer than $32$ nodes, rotations do not make $\sC$ increase anyway.

Since the sum $\sC$ is initially zero, and terminates with a non-negative value, no more than $\O(k / \delta)$ rotations are performed on trees with $32$ nodes or more, and no more than $\O(k)$ rotations can be performed on trees with fewer than $32$ nodes.
\end{proof}

In particular, let us focus on some approach to the critical point~$(\alpha_\scc,\beta_\scc) = (1/\sqrt{2},\scrB(1/\sqrt{2}))$.
When setting~$\alpha = \alpha_\scc + x$ and~$\beta = \beta_\scc + x$, and provided that~$0 < x < 1/10$, the inequalities~$1/\sqrt{2} < \alpha < 3/4$ and~$\scrB(\alpha) < \beta < \alpha^2$ are valid.
Moreover,
\begin{itemize}[itemsep=0.5pt,topsep=0.5pt]
\item Theorem~\ref{thm:height} states that~$\GCB{\alpha,\beta}$-trees with~$\bN$ nodes are of height
\[h \leqslant - 2 \log_2(\bN+1)/\log_2(\beta_\scc) + 7 \log_2(\bN+1) x;\]

\item Theorem~\ref{thm:IPL} states that~$\GCB{\alpha,\beta}$-trees with~$\bN$ nodes are of external path length
\[\lambda \leqslant (\bN+1) \log_2(\bN+1) / \Delta_\scc + 2 (\bN+1) \log_2(\bN+1) x,\]
where~$\Delta_\scc = (\sH_2(\alpha_\scc) + \alpha_\scc \sH_2(\beta_\scc / \alpha_\scc)) / (1 + \alpha_\scc)$;

\item Theorem~\ref{thm:BU-complexity} states that inserting or deleting~$k$ leaves in~$\GCB{\alpha,\beta}$-trees requires~$\O(k/x)$ rotations.
More precisely, digging in the constants hidden in the~$\O$ notation yields a (crude) upper bound of less than~$30 (1 + 28/x) k$ rotations.
\end{itemize}

\bibliography{weight}

\clearpage

\appendix

\section{Missing proofs of Sections~\ref{sec:grandchildren} and~\ref{sec:BU}}

\ipl*

\begin{proof}
Let~$\lambda(\T)$ be the sum of the weights of the nodes of a tree~$\T$: we prove by induction on~$|\T|$ that~$\lambda(\T) \leqslant |\T| \log_2(|\T|) / \Delta$.

When~$|\T| = 1$, this inequality rewrites as~$0 \leqslant 0$.
When~$|\T| = 2$, it rewrites as~$\Delta \leqslant 1$, which follows from the fact that~$\sH_2(x) \leqslant 1$ whenever~$x \in [0,1]$.
Therefore, we assume that~$|\T| \geqslant 3$.

Let~$n$ be the root of~$\T$, and let~$\T_1$ and~$\T_2$ be its left and right sub-trees.
Without loss of generality, we assume that~$|\T_1| \geqslant |\T_2|$.
Thus,~$|\T_1| \geqslant |\T|/2 > 1$, and~$\T_1$ is non-empty.
Let~$n_1$ be its root, and let~$\T_{11}$ and~$\T_{12}$ be the left and right sub-trees of~$n_1$.
Once again, without loss of generality, we assume that~$|\T_{11}| \geqslant |T_{12}|$.

Let~$t = |\T|$,~$x = |\T_1| / |\T|$ and~$y = |\T_{11}| / |\T_1|$.
The induction hypothesis states now that
\begin{align*}
\Delta \lambda(\T) / t
& = \Delta (x+1) + \Delta \lambda(\T_{11})/t + \Delta \lambda(\T_{12})/t + \Delta \lambda(\T_2)/t \\
& \leqslant \Delta (x+1) + x y \log_2(t x y) + x (1-y) \log_2(t x (1-y)) + (1-x) \log_2(t (1-x)) \\
& \leqslant \Delta (x+1) + \log_2(t) - x \sH_2(y) - \sH_2(x),
\end{align*}
and it remains to prove that the quantity~$\sF_{\alpha,\beta}(x,y) = x \sH_2(y) + \sH_2(x) - \Delta (x+1)$ is non-negative.

First, let~$\gamma = \beta / \alpha$.
If~$\gamma \leqslant x \leqslant \alpha$, since~$\sH_2$ is decreasing on the interval~$[1/2,1]$ and~$1/2 \leqslant y \leqslant \beta / x$, we know that~$\sF_{\alpha,\beta}(x,y) \geqslant \sF_{\alpha,\beta}(x,\beta / x)$.
Moreover, the function~$\sG \colon x \mapsto \sF_{\alpha,\beta}(x,\beta / x)$ is concave on~$[\gamma,\alpha]$, because its second derivative is
\[\sG''(x) = - \frac{1-\beta}{(x - \beta)(1-x) \ln(2)} < 0\]
whenever~$\beta \leqslant \gamma < x < \alpha \leqslant 1$.
Observing that~$\sG(\alpha) = \sF_{\alpha,\beta}(\alpha,\gamma) = 0$ and that
\begin{align*}
\sG(\gamma) & = \sF_{\alpha,\beta}(\gamma,\alpha) = (\beta \sH_2(\alpha) + \sH_2(\gamma)) - (\gamma+1)(\sH_2(\alpha) + \alpha \sH_2(\gamma))/(\alpha+1) \\
& = (1- \beta)(\sH_2(\gamma)-\sH_2(\alpha))/(\alpha+1) \geqslant 0,
\end{align*}
we conclude that~$\sF_{\alpha,\beta}(x,y) \geqslant \sG(x) \geqslant \min\{\sG(\alpha),\sG(\gamma)\} \geqslant 0$ whenever~$\gamma \leqslant x \leqslant \alpha$.
Finally, if~$1/2 \leqslant x \leqslant \gamma$, and since~$1/2 \leqslant y \leqslant \alpha$, we have~$\sF_{\alpha,\beta}(x,y) \geqslant \sF_{\alpha,x}(x,y) \geqslant \sF_{\alpha,x}(x,\alpha) \geqslant 0$.
\end{proof}

\tightipl*

\begin{proof}
Let~$\gamma = \beta / \alpha$.
For each integer~$s \geqslant 1$, we define inductively a tree~$\T(s)$ of weight~$s$ as follows:
\begin{itemize}[itemsep=0.5pt,topsep=0.5pt]
\item if~$s \leqslant 2$,~$\T(s)$ is the only tree of weight~$s$;
\item if~$s \geqslant 3$,~$\T(s)$ is the tree whose left child is~$\T(s_1)$ and whose right child's children are~$\T(s_{21})$ and~$\T(s_{22})$, where~$s_1 = s - \lfloor \alpha s\rfloor$,~$s_{21} = \lfloor \gamma \lfloor \alpha s \rfloor \rfloor$ and~$s_{22} = s - s_1 - s_{21}$.
\end{itemize}

We present in Figure~\ref{fig:T(s)} the trees~$\T(s)$ obtained when~$1 \leqslant s \leqslant 6$.
Since~$2/3 < \gamma \leqslant \alpha < 3/4$, these trees do not depend on the values of~$\alpha$ and~$\gamma$ (or~$\beta$).

\begin{figure}[t]
\begin{center}
\begin{tikzpicture}[scale=0.71]
\draw[thick] (5,1) -- (6,0)
(8,0) -- (9,1) -- (10,0)
(12,1) -- (13,2) -- (14,1) -- (13,0)
(16,1) -- (17,2) -- (19,0) (17,0) -- (18,1);
\foreach \x/\y in {3/0,6/0,8/0,10/0,12/1,13/0,16/1,17/0,19/0}{
 \foreach \t in {-0.5,0.5}{
 \draw[thick,draw=black!50] (\x,\y) --++ (\t,-0.5);
  \draw[fill=white,draw=white] (\x+\t,\y-0.5) circle (0.15);
  \node at (\x+\t,\y-0.5) {$\bot$};
}}
\draw[thick,draw=black!50] (5,1) --++ (-0.5,-0.5) (14,1) --++ (0.5,-0.5);
\foreach \x/\y in {1/0,4.5/1,14.5/1}{
  \draw[fill=white,draw=white] (\x,\y-0.5) circle (0.15);
  \node at (\x,\y-0.5) {$\bot$};
}
\foreach \x/\y/\l in {3/0/2,%
5/1/3,6/0/2,%
8/0/2,9/1/4,10/0/2,%
12/1/2,13/2/5,14/1/3,13/0/2,%
16/1/2,17/2/6,18/1/4,17/0/2,19/0/2}{
 \draw[fill=white,thick] (\x,\y) circle (0.3);
 \node at (\x,\y) {\l};
}
\foreach [count=\c] \x in {1,3,5.5,9,13,17.5}{
 \node[anchor=north] at (\x,-1) {$\T(\c)$};
}
\end{tikzpicture}
\vspace{-1.2em}
\end{center}
\caption{Trees~$\T(1)$ to~$\T(6)$.
Nodes are labelled by their weight; empty trees are denoted by~$\bot$.
\label{fig:T(s)}}
\end{figure}

We first prove by induction on~$s$ that~$\T(s)$ is a~$\GCB{\alpha,\beta}$-tree.
This is visibly true when~$s \leqslant 4$, hence we assume that~$s \geqslant 5$.
Let~$s_2 = \lfloor \alpha s \rfloor$ be the weight of the right child of~$\T(s)$: it suffices to prove that~$s_1 \leqslant \gamma s$,~$s_2 \leqslant \alpha s$,~$s_{21} \leqslant \gamma s_2$ and~$s_{22} \leqslant \gamma s_2$.
The inequalities~$s_2 \leqslant \alpha s$ and~$s_{21} \leqslant \gamma s_2$ are immediate, and~$s_2 \geqslant \lfloor 2s/3 \rfloor \geqslant 3$, thereby proving that:
\begin{itemize}
\item $s_1 \leqslant (1-\alpha+1/s) s \leqslant (1-2/3+1/5)s < 2s/3 < \gamma s$;
\item $s_{22} \leqslant (1-\gamma+1/s_2) s_2 \leqslant (1-2/3+1/3) s_2 = 2s_2/3 < \gamma s_2$.
\end{itemize}

Now, let~$h(s)$ be the height of~$\T(s)$ and let~$\lambda(s)$ be the sum of the weights of the nodes of~$\T(s)$.
We will prove inductively that~$h(s) \geqslant h^+(s+\kappa)$ and~$\lambda(s) \geqslant \lambda^+(s+\kappa) + v$ for all~$s \geqslant 1$, where we set~$h^+(x) = - 2 \log_\beta(x) - 7$,~$\lambda^+(x) = x \log_2(x)/\Delta - 4 x$,~$\kappa = (\gamma+1)/(1-\beta)$ and~$v = (1 + (\alpha+1)\kappa)/2$.

First, we observe, when~$(\alpha,\beta) \in \D'$, that~$4/5 \leqslant \sH_2(3/4) \leqslant \Delta \leqslant \sH_2(2/3) \leqslant \log_2(5 e)/4$, {}$3 \leqslant \kappa \leqslant 4$ and~$3 \leqslant v \leqslant 4$.
We will use these inequalities several times below.

For example,~$\kappa+1 \leqslant 5 \leqslant (4/3)^6 \leqslant \beta^3$ and~$\kappa+2 \leqslant 6 \leqslant (4/3)^7 \leqslant \beta^{7/2}$, so that~$h(s) = s-2 \geqslant h^+(s+\kappa)$ when~$s = 1$ or~$s = 2$.

Then, if~$s \geqslant 3$, observing that~$s_{21}+\kappa \geqslant \beta s - \gamma - 1 = \beta(s+\kappa)$ and using the induction hypothesis proves that~$h(s) \geqslant h(s_{21}) + 2 \geqslant h^+(s_{21}+\kappa) + 2 \geqslant 2-2 \log_\beta(\beta(s+\kappa))-7 = h^+(s+\kappa)$.

Second, let~$\lambda^-(s)$ be the smallest possible sum of the weights of the nodes of a binary tree with weight~$s$.
By construction,~$\lambda(s) \geqslant \lambda^-(s)$, and~$\lambda^-$ shines as the non-decreasing convex function defined by~$\lambda^-(1) = 0$ and~$\lambda^-(s) = s + \lambda^-(\lfloor s/2 \rfloor) + \lambda^-(\lceil s/2 \rceil)$ for all~$s \geqslant 2$.
Consequently, observing that
\[\lambda^+(s+\kappa) + v \leqslant 5 (s+\kappa) \log_2(s+\kappa)/4 - 4 (s+\kappa) + v \leqslant 5(s+4) \log_2(s+4)/4 - 4(s+3) + 4\]
whenever~$s \geqslant 1$ suffices to check (by computer) that~$\lambda(s) \geqslant \lambda^-(s) \geqslant s \lfloor \log_2(s) \rfloor \geqslant \lambda^+(s+\kappa)$ for all integers~$s \leqslant 26$.

Then, if~$s \geqslant 27$, observe that~$\alpha(1-\alpha)(s+\kappa) \geqslant 2/3 \times 1/4 \times 30 \geqslant 5$.
It follows that:
\begin{itemize}
\item $s_1 + \kappa \geqslant (1-\alpha)s + \kappa \geqslant (1-\alpha)(s+\kappa) \geqslant 5$;
\item $s_{21} + \kappa \geqslant \alpha \gamma s + \kappa - \gamma - 1 = \alpha \gamma (s+\kappa) \geqslant 5$;
\item $s_{22} + \kappa \geqslant \alpha(1-\gamma)s + \kappa +\gamma - 1 \geqslant \alpha(1-\gamma)(s+\kappa) \geqslant 5$.
\end{itemize}

Moreover, the function~$\lambda^+$ is increasing on the interval~$(16^\Delta/e,+\infty)$, and~$16^\Delta/e \leqslant 5$.
In addition, the equality~$\lambda^+(x) = (\alpha+1)x + \lambda^+((1-\alpha)x) + \lambda^+(\alpha \gamma x) + \lambda^+(\alpha(1-\gamma)x$ is valid for all~$x \geqslant 1$.
Thus, the induction hypothesis proves that
{\allowdisplaybreaks
\begin{align*}
\lambda(s) & \geqslant s + s_2 + \lambda(s_1) + \lambda(s_{21}) + \lambda(s_{22}) \\
& \geqslant s + s_2 + \lambda^+(s_1+\kappa) + \lambda^+(s_{21}+\kappa) + \lambda^+(s_{22}+\kappa) + 3 v \\
& \geqslant (\alpha+1)s + \lambda^+((1-\alpha)(s+\kappa)) + \lambda^+(\alpha \gamma(s+\kappa)) + \lambda^+(\alpha(1-\gamma)(s+\kappa)) + (3v-1)\\
& \geqslant \lambda^+(s+\kappa) - (\alpha+1)\kappa + (3v-1) = \lambda^+(s+\kappa) + v.
\end{align*}}

Finally, we check by hand that~$\lambda(s) \geqslant s \lfloor \log_2(s) \rfloor \geqslant 5 s \log_2(s) / 4 - 4 s + 4 \geqslant  s \log_2(s) / \Delta - 4 s + 4$ when~$1 \leqslant s \leqslant 8$, whereas, since~$\lambda^+$ is increasing on~$[5,+\infty)$,
\begin{align*}
\lambda(s) & \geqslant \lambda^+(s+\kappa)+v \geqslant \lambda^+(s+3)+3
= (s+3) \log_2(s+3) / \Delta - 4s - 9 \\
& \geqslant (s \log_2(s) + 3 \log_2(s+3)) / \Delta - 4s - 9 \\
& \geqslant s \log_2(s) / \Delta + 15 \log_2(12) / 4 - 4s - 9
\geqslant s \log_2(s) / \Delta - 4 s + 4
\end{align*}
when~$s \geqslant 9$.
\end{proof}

\bbalancing*

\begin{proof}
Below, we will use the following inequalities, which are all easy to check with any computer algebra system (whereas some look too frightening to check by hand) whenever~$1/\sqrt{2} \leqslant \alpha < 3/4$ and~$\scrB(\alpha) \leqslant \beta \leqslant \alpha^2$.
Except the first three inequalities, each of them is labelled and later reused to prove another inequality with the same label;
the first two inequalities already prove that~$1/4 < \betaa \leqslant \beta$, as stated in Lemma~\ref{lem:B-balancing}.
\newcommand{\vjline}{\vphantom{(\alpha^3\hat{\gamma})}}
\vspace{-6mm}

\begin{multicols}{2}
\begin{align}
& \vjline 1/4 < \smash{\betaa}; \notag \\
& \vjline \alpha < 2 \smash{\betaa}; \notag \\
& \vjline (1-\delta)(\alpha-\smash{\betaa}) < \delta(1-\alpha); \tag{3.2} \\
& \vjline 1-\smash{\betaa} < \delta; \tag{3.4} \\
& \vjline (1-\delta)\beta < \delta(1-\alpha); \tag{4.2} \\
& \vjline (1-\delta)\alpha\beta < \delta(1-\alpha); \tag{4.4} \\
& \vjline 1-\delta+(\delta+\beta\delta-1)\alpha \leqslant \delta\smash{\betaa}; \tag{4.5\smash{\textsuperscript{b}}}
\end{align}

\begin{align}
& \vjline \text{\rlap{$\smash{\betaa} \leqslant \beta;$}}\hphantom{1-\hat{\gamma}-\alpha\hat{\gamma}^2+(\hat{\gamma}+\hat{\gamma}\alpha\gamma-1)\alpha \leqslant 0;} \notag \\
& \vjline \alpha^2 < \delta; \tag{3.1} \\
& \vjline (1-\delta)(1-\hat{\gamma}) \leqslant \delta\hat{\gamma}(1-\alpha); \tag{3.3} \\
& \vjline (1-\delta)(\alpha-\smash{\betaa}) < \delta(\smash{\betaa}-\alpha\beta); \tag{4.1} \\
& \vjline \alpha-\smash{\betaa}+\alpha\beta < \delta; \tag{4.3} \\
& \vjline 1 \leqslant \delta(1+\beta); \tag{4.5\smash{\textsuperscript{a}}} \\
& \vjline 1+(\beta-1)\hat{\gamma} < \delta. \tag{4.6}
\end{align}
\end{multicols}

Then, let~$|m|$ denote the weight of a node~$u$ in the tree~$\T$, and let~$|m|'$ denote its weight in~$\T'$.
When these weights are equal, we will prefer using the notation~$|m|$ even when considering the weight of~$m$ in~$\T'$.
The root of~$\T$ is denoted by~$n$, its left and right children are denoted by~$n_1$ and~$n_2$, and so on.
Finally, we set~$\hat{\gamma} = \betaa/\alpha$, so that~$\delta = \min\{\alphaa,\hat{\gamma}\}$.

Since~$2 \beta \geqslant \alpha$ and~$n$ is~$\cbal{\alpha}$, observing that~$2 \beta |n| \geqslant \alpha |n| \geqslant |n_1| = |n_{11}| + |n_{12}|$ proves that the inequalities~$|n_{11}| > \beta|n|$ and~$|n_{12}| > \beta|n|$ are incompatible.
Similarly, the nodes~$n_1$ and~$n_2$ are~$\cbal{\alpha}$, and thus observing that~$2\beta|n| \geqslant \alpha |n| = \alpha|n_1|+\alpha|n_2| \geqslant |n_{11}|+|n_{21}|$ proves that the inequalities~$|n_{11}| > \beta|n|$ and~$|n_{21}| > \beta|n|$ are incompatible.
Finally, for symmetry reasons, all four inequalities~$|n_{ij}| > \beta|n|$ are incompatible:
this makes our description of the algorithm unambiguous, as announced.

Then, proving that the desired rotations can indeed be performed simply requires showing that, if some case~1 to~4 happens, the grandchild~$n_{ij}$ be a grandchild of~$n$ whose weight~$|n_{ij}|$ is maximal is an actual tree node instead of an empty node of weight~$1/2$ or less.
If this were the case, we would have~$|n_i| = 2 |n_{ij}| \leqslant 1$, and~$n_i$ would be either an empty node or a leaf; in both cases, we would have~$|n| \geqslant 2 |n_i|$, i.e.,~$|n| \geqslant 4|n_{ij}|$, contradicting the inequality~$|n_{ij}| > y |n| \geqslant \betaa|n| > |n|/4$.

It remains to prove that each affected node~$m$ is~$\cbal{\delta}$; its children being~$\cbal{\alpha}$, it will then be~$\gcbal{\betaa}$.
Furthermore, if a rotation was triggered,~$n$ was not~$\gcbal{\beta}$ but its children were~$\cbal{\alpha}$, which proves that~$\balc(\T) > \beta / \alpha \geqslant \hat{\delta} \geqslant \balc(\T')$.

When~$|n_{11}| > w |n| \geqslant \betaa |n|$, a simple rotation is performed, and
\begin{align}
|n_{11}| & \leqslant \alpha |n_1| \leqslant \alpha^2 |n| \leqslant \delta |n| = \delta |n_1|'; \tag{3.1} \\
(1-\delta) |n_{12}| & = (1-\delta)(|n_1|-|n_{11}|) \leqslant (1-\delta)(\alpha |n|-\smash{\betaa} |n|) \notag \\
& \leqslant \delta(1-\alpha)|n| \leqslant \delta(|n|-|n_1|) = \delta |n_2|; \tag{3.2} \\
(1-\delta) |n_2| & = (1-\delta)(|n| - |n_1|) \leqslant (1-\delta)(|n| - |n_{11}|/\alpha) \leqslant (1-\delta)(1-\hat{\gamma})|n| \notag \\
& \leqslant \delta \hat{\gamma} (1-\alpha) |n| \leqslant \delta(1/\alpha-1)|n_{11}| < \delta(|n_1| - |n_{11}|) = \delta |n_{12}|; \tag{3.3} \\
|n| & = |n| - |n_{11}| \leqslant (1-\smash{\betaa}) |n| \leqslant \delta |n| = \delta |n_1|', \tag{3.4}
\end{align}
which means precisely that~$n$ and~$n_1$ are~$\cbal{\delta}$ in~$\T'$.

Similarly, when~$|n_{12}| > w |n| \geqslant \betaa |n|$, a double rotation is performed, and
\begin{align}
(1-\delta)|n_{11}| & = (1-\delta)(|n_1| - |n_{12}|) \leqslant (1-\delta)(\alpha |n| - \smash{\betaa} |n|) \notag \\
& \leqslant \delta(\smash{\betaa} |n|-\alpha \beta |n|) \leqslant \delta(|n_{12}| - \beta |n_1|) \leqslant \delta(|n_{12}| - |n_{122}|) = \delta |n_{121}|; \tag{4.1} \\
(1-\delta) |n_{121}| & \leqslant (1-\delta)\beta |n_1| \leqslant \delta(|n_1|-\alpha |n_1|) \leqslant \delta(|n_1|-|n_{12}|) = \delta |n_{11}|; \tag{4.2} \\
|n_1|' & = |n_1| - |n_{12}| + |n_{121}| \leqslant |n_1| - |n_{12}| + \beta |n_1| \notag \\
& \leqslant (\alpha-\smash{\betaa}+\alpha\beta)|n| \leqslant \delta |n| = \delta |n_{12}|'; \tag{4.3} \\
(1-\delta) |n_{122}| & \leqslant (1-\delta) \beta |n_1| \leqslant (1-\delta) \alpha\beta |n| \leqslant \delta (1-\alpha)|n| \leqslant \delta (|n| - |n_1|) = \delta |n_2|; \tag{4.4} \\
(1-\delta) |n_2| & \leqslant (1-\delta)|n_2| + \beta \delta |n_1| - \delta |n_{121}| = (1-\delta)|n| + (\delta+\beta\delta-1)|n_1| - \delta |n_{121}| \notag \\
& \leqslant (1-\delta)|n| + (\delta+\beta\delta-1)\alpha|n| - \delta |n_{121}| \tag{4.5\smash{\textsuperscript{a}}} \\
& \leqslant \delta\smash{\betaa}|n| - \delta |n_{121}| \leqslant \delta|n_{12}| - \delta|n_{121}| = \delta|n_{122}|; \tag{4.5\smash{\textsuperscript{b}}}\\
|n|' & = |n| - |n_1| + |n_{122}| \leqslant |n| + (\beta - 1) |n_1| \notag \\
& \leqslant |n| + (\beta - 1) |n_{12}| / \alpha \leqslant |n| + (\beta - 1) \hat{\gamma} |n| \leqslant \delta|n| = \delta |n_{12}|', \tag{4.6}
\end{align}
which means precisely that~$n$,~$n_1$ and~$n_{12}$ are~$\cbal{\delta}$ in~$\T'$.

Finally, the cases where~$|n_{21}| > w |n|$ or~$|n_{22}| > w |n|$ are symmetrical, and therefore the conclusion of Lemma~\ref{lem:B-balancing} is valid in these cases too.
\end{proof}

\section{Proofs of Section~\ref{sec:TD}}\label{app:td}

Like in Section~\ref{sec:BU}, the proof of Proposition~\ref{pro:RAB-balancing} is split into three parts, the first two of which concern Algorithms~\ref{routine:A} and~\ref{routine:B} and heavily rely on algebraic computations better checked on a computer algebra system.

\begin{lemma}\label{lem:RA-balancing}
Let~$\alpha$ be a real number such that~$1 / \sqrt{2} \leqslant \alpha < 3/4$.
When~$\alphaa \leqslant u \leqslant \alpha$, the five cases in which the~$\sR\sA(u,\alphaa,\sR(u),\sR(\alphaa))$-balancing algorithm consists are pairwise incompatible when operating on trees with~$30$ nodes or more, and those rotations they trigger can always be performed;
the algorithm itself transforms each~$\GCB{\alpha}$-tree~$\T$ with at least~$30$ nodes and whose root is not robustly~$\cbal{u}$ into a tree~$\T'$ such that~$\balc(\T') \leqslant \balc(\T)$, whose root is robustly~$\cbal{\alphaa}$ and whose sub-trees are~$\GCBB{\alphab}{\alpha}$-trees.
\end{lemma}

\begin{proof}
Below, we will often use the fact that~$(1-\alpha) |x| \leqslant |y| \leqslant \alpha |x|$ whenever~$y$ is a child of an~$\cbal{\alpha}$ node~$x$.
Like in the proofs of Propositions~\ref{lem:A-balancing} and~\ref{lem:B-balancing}, we will also use the following inequalities, which can be checked with any computer algebra system.
Most of these inequalities involve an integer parameter~$k$ and are valid whenever~$1/\sqrt{2} \leqslant \alpha < 3/4$ and~$k$ is large enough, the lower bound on~$k$ being indicated each time.
Each of them is labelled and later reused to prove another inequality with the same label; inequality~(5.2\,+\,6.2\,+\,6.4) is used to prove the three subsequent inequalities~(5.2),~(6.2) and~(6.4); inequality~(6.1\,+\,6.5) is used to prove both subsequent inequalities~(6.1) and~(6.5).
\vspace{-4pt}
\newcommand{\vjline}{\vphantom{(k\alpha^2\alphaa u)}}
\noindent%
\begin{align}
& \vjline k < 2 \alphaa k + 2 \sR(\alphaa) & \text{when } k \geqslant 4; \notag \\
& \vjline \alpha^2 k < \alphaa k + \sR(\alphaa) & \text{when } k \geqslant 5; \tag{5.1} \\
& \vjline (1-\alphab)\alpha^2 < \alphab(1-\alpha); \tag{5.2\,+\,6.2\,+\,6.4} \\
& \vjline (1-\alphab)((1-\alphaa)k-\sR(\alphaa)) < \alphab(1-\alpha)(\alphaa k + \sR(\alphaa)) & \text{when } k \geqslant 4; &  \tag{5.3} \\
& \vjline (1-\alphab)((1-\alphaa)k - \sR(\alphaa)) < \alphab(1-\alpha)((2\alphaa-1)k + 2 \sR(\alphaa)) & \text{when } k \geqslant 16; \tag{6.1\,+\,6.5} \\
& \vjline (1-\alpha)((1-\alphaa)k-\sR(\alphaa)) + \alpha^2 k < \alphaa k + \sR(\alphaa) & \text{when } k \geqslant 11; \tag{6.3} \\
& \vjline k+(\alpha^2-1)(\alphaa k + \sR(\alphaa)) < \alphaa k + \sR(\alphaa) & \text{when } k \geqslant 31. \tag{6.6}
\end{align}

In particular, when~$u \geqslant \alphaa$, we have~$2 u|n| + 2 \sR(u) \geqslant 2 \alphaa |n| + 2 \sR(\alphaa) > |n| = |n_1| + |n_2|$, which makes the inequalities~$|n_1| > u|n| + \sR(u)$ and~$|n_2| > u |n| + \sR(u)$ incompatible.
Second, in cases~1 and~3, we have~$|n_1| > u |n| + \sR(u) \geqslant 31/\sqrt{2} - 1 > 20$, which proves that~$n_1$ is an actual tree node; in case~3, we also have
\[|n_{12}| = |n_1| - |n_{11}| > u|n| + \sR(u) - (1-\alphaa)|n| + \sR(\alphaa) \geqslant (\sqrt{2}-1)\times 31-2 > 10,\]
and~$n_{12}$ is also a tree node.
Hence, those rotations that the algorithm triggers can always be performed.

When~$\alpha|n| \geqslant |n_1| > u |n| + \sR(u) \geqslant \alphaa |n| + \sR(\alphaa)$ and~$|n_{11}| \geqslant (1-\alphaa) |n| - \sR(\alphaa)$, a simple rotation is performed, and
\begin{align}
|n_{11}| & \leqslant \alpha |n_1| \leqslant \alpha^2 |n| \leqslant \alphaa |n| + \sR(\alphaa) = \alphaa |n_1|' + \sR(\alphaa); \tag{5.1} \\
(1-\alphab) |n_{12}| & \leqslant (1-\alphab)\alpha |n_1| \leqslant (1-\alphab)\alpha^2 |n| \leqslant \alphab(1-\alpha) |n| \leqslant \alphab |n_2|; \tag{5.2} \\
(1-\alphab)|n_2| & = (1-\alphab)(|n|-|n_1|) \leqslant (1-\alphab)((1-\alphaa)|n| - \sR(\alphaa)) \notag \\
& \leqslant \alphab(1-\alpha)(\alphaa|n| + \sR(\alphaa)) \leqslant \alphab(1-\alpha) |n_1| \leqslant \alphab|n_{12}|; \tag{5.3} \\
|n|' & = |n| - |n_{11}| \leqslant \alphaa |n| + \sR(\alphaa), \tag{5.4}
\end{align}
which means that~$n_1$ is robustly~$\cbal{\alphaa}$ in~$\T'$ and~$n$ is~$\cbal{\alphab}$ in~$\T'$.

Similarly, when~$\alpha |n| \geqslant |n_1| > u |n| + \sR(u) \geqslant \alphaa |n| + \sR(\alphaa)$ and~$|n_{11}| < (1-\alphaa) |n| - \sR(\alphaa)$, a double rotation is performed, and
\begin{align}
(1-\alphab) |n_{11}| & \leqslant (1-\alphab)((1-\alphaa)|n| - \sR(\alphaa)) \leqslant \alphab (1-\alpha)((2\alphaa-1)|n| + 2 \sR(\alphaa)) \notag \\
& \leqslant \alphab(1-\alpha)(|n_1|-|n_{11}|) = \alphab (1-\alpha) |n_{12}| \leqslant \alphab|n_{121}|; \tag{6.1} \\
(1-\alphab) |n_{121}| & \leqslant (1-\alphab) \alpha |n_{12}| \leqslant (1-\alphab) \alpha^2 |n_1| \leqslant \alphab (1-\alpha) |n_1| \leqslant \alphab |n_{11}|; \tag{6.2} \\
|n_1|' & = |n_{11}| + |n_{121}| \leqslant |n_{11}| + \alpha |n_{12}| = (1-\alpha)|n_{11}| + \alpha |n_1| \notag \\
& \leqslant (1-\alpha)((1-\alphaa)|n| - \sR(\alphaa)) + \alpha^2 |n| \notag \\
& \leqslant \alphaa |n| + \sR(\alphaa) = \alphaa |n_{12}|' + \sR(\alphaa); \tag{6.3} \\
(1-\alphab)|n_{122}| & \leqslant (1-\alphab)|n_{12}| \leqslant (1-\alphab)\alpha|n_1| \leqslant (1-\alphab)\alpha^2 |n| \notag \\
& \leqslant \alphab(1-\alpha) |n| \leqslant \alphab |n_2|; \tag{6.4} \\
(1-\alphab)|n_2| & = (1-\alphab)(|n|-|n_1|) \leqslant (1-\alphab)((1-\alphaa)|n|-\sR(\alphaa)) \notag \\
& \leqslant \alphab (1-\alpha)((2\alphaa-1)|n|+2\sR(\alphaa)) \notag \\
& \leqslant \alphab(1-\alpha)(|n_1|-|n_{11}|) = \alphab (1-\alpha) |n_{12}| \leqslant \alphab|n_{122}|; \tag{6.5} \\
|n|' & = |n|-|n_1|+|n_{122}| \leqslant |n|-|n_1|+\alpha|n_{12}| \leqslant |n| + (\alpha^2-1)|n_1| \notag \\
& \leqslant |n|+(\alpha^2-1)(\alphaa|n|+\sR(\alphaa)) \leqslant \alphaa |n| + \sR(\alphaa) =  \alphaa |n_{12}|' + \sR(\alphaa), \tag{6.6}
\end{align}
which means that~$n_{12}$ is robustly~$\cbal{\alphaa}$ and~$n$ and~$n_1$ are~$\cbal{\alphab}$ in~$\T'$.

Finally, the cases where~$|n_2| > \alpha |n| + \sR(\alpha)$ are symmetrical, and thus the conclusion of Lemma~\ref{lem:RA-balancing} is valid in these cases too.
\end{proof}

\begin{lemma}\label{lem:RB-balancing}
Let~$\alpha$ and~$\beta$ be real numbers such that~$1/\sqrt{2} \leqslant \alpha < 3/4$ and~$\scrB(\alpha) \leqslant \beta \leqslant \alpha^2$.
When~$\betaa \leqslant y \leqslant \beta$, the five cases in which the~$\sR\sB(y,\sR(y))$-balancing algorithm consists are pairwise incompatible when operating on~$\GCB{\alpha}$-trees with at least~$23$ nodes, and those rotations they trigger can always be performed;
the algorithm itself transforms each~$\GCBB{\alpha,1}{\alpha,\beta}$-tree~$\T$ with~$22$ nodes or more and whose root is not robustly~$\gcbal{y}$ into a tree~$\T'$ such that~$\balc(\T') \leqslant \balc(\T)$, whose root is robustly~$(\alphaa,\betaa)$-balanced and whose sub-trees are~$\GCBB{\alphab,1}{\alpha,\beta}$-trees.
\end{lemma}

\begin{proof}
Like in the proof of Lemma~\ref{lem:RA-balancing}, we will often use the fact that~$(1-\alpha) |x| \leqslant |y| \leqslant \alpha |x|$ whenever~$y$ is a child of an~$\cbal{\alpha}$ node~$x$;
we will also the following inequalities, which can be checked with any computer algebra system, and most of which involve a large enough integer parameter~$k$.
Except the first one, each of them is labelled and later reused to prove another inequality with the same label.

Each label bears a superscripted tag a, b or c: this is because, unlike for Lemma~\ref{lem:B-balancing}, showing that a node is~$\cbal{(\beta/\alpha)}$ is not enough to complete our proofs.
Instead, we will always compare separately the weight of a node~$m$ with the weights of its parent (in which case we use the tag~a) or of its grand-parent (in which case we use the tag~b) in~$\T'$, or with the weight of the node~$n_1$ in~$\T$ (in which case we use a tag~c).
This will be used to prove that (a)~the parent of~$m$ is \cbal{}, or (b)~the grand-parent of~$m$ is well \gcbal{}, or (c)~$\balc(\T') \leqslant \balc(\T)$;
the inequality~($\star$) is simply reused to prove the subsequent inequality~(7.3\textsuperscript{b}\,+\,8.1\textsuperscript{b}).
\newcommand{\vjline}{\vphantom{(k\alpha^2\alphaa u)}}
{\allowdisplaybreaks
\begin{align}
& \vjline \alpha k < 2 \betaa k + 2 \sR(\betaa) & \text{when } k \geqslant 5; \tag{$\star$} \\
& \vjline \alpha \beta k < \betaa k+\sR(\betaa) & \text{when } k \geqslant 5; \tag{7.1\textsuperscript{b}\,+\,8.2\textsuperscript{b}\,+\,8.4\textsuperscript{b}} \\
& \vjline \alpha^2 k < \alphaa k + \sR(\alphaa) & \text{when } k \geqslant 5; \tag{7.2\textsuperscript{a}} \\
& \vjline (1-\alphab) \alpha^2 < \alphab(1-\alpha); & \tag{7.3\textsuperscript{a}\,+\,8.2\textsuperscript{a}\,+\,8.4\textsuperscript{a}} \\
& \vjline (1-\alphab)(\alpha k - \sR(\betaa) - \betaa k) < \alphab \betaa (1-\alpha) k & \text{when } k \geqslant 3; \tag{7.4\textsuperscript{a}} \\
& \vjline \alpha k < (\alpha+1)(\betaa k + \sR(\betaa)) & \text{when } k \geqslant 9; \tag{7.4\textsuperscript{b}\,+\,7.5\textsuperscript{c}\,+\,8.5\textsuperscript{b}} \\
& \vjline (1-\betaa) k-\sR(\betaa) < \alphaa k + \sR(\alphaa) & \text{when } k \geqslant 7; \tag{7.5\textsuperscript{a}} \\
& \vjline (1-\alphab)(\alpha k - \betaa k -\sR(\betaa)) < \alphab(1-\alpha)(\betaa k+\sR(\betaa)) & \text{when } k \geqslant 5; \tag{8.1\textsuperscript{a}} \\
& \vjline \alpha k - (1-\alpha)(\betaa k + \sR(\betaa)) < \alphaa k + \sR(\alphaa) & \text{when } k \geqslant 9; \tag{8.3\textsuperscript{a}} \\
& \vjline (1-\alphab)(k - (\beta k + \sR(\beta))/\alpha) < \alphab (1-\alpha)(\beta k + \sR(\beta)) & \text{when } k \geqslant 9;  \tag{8.5\textsuperscript{a}} \\
& \vjline \alpha k - (1-\beta)(\betaa k+ R(\betaa)) < \alpha (\alphaa k + \sR(\alphaa)) & \text{when } k \geqslant 2\rlap{3;}\phantom{;} \tag{8.6\textsuperscript{a}} \\
& \vjline k < u k + v & \text{when } k \geqslant 5, \tag{8.6\textsuperscript{c}}
\end{align}}
where~$u = (2-\beta)\betaa/\alpha$ and~$v = (2-\beta)\sR(\betaa)/\alpha + 1$. 

In particular, since~$n$ is~$\cbal{\alpha}$, inequality~($\star$) indicates that
\begin{align}
2 y |n| + 2 \sR(y) \geqslant 2 \betaa |n| + 2 \sR(\betaa) \geqslant \alpha |n| \geqslant |n_1| = |n_{11}| + |n_{12}|, \tag{7.3\textsuperscript{b}\,+\,8.1\textsuperscript{b}}
\end{align}
which makes the inequalities~$|n_{11}| > y|n| + \sR(y)$ and~$|n_{12}| > y|n| + \sR(y)$ incompatible.
Similarly,~$n_1$ and~$n_2$ are~$\cbal{\alpha}$, and thus
\[
2y |n| + 2 \sR(y) \geqslant \alpha |n| = \alpha|n_1|+\alpha|n_2| \geqslant |n_{11}|+|n_{21}|,\]
which proves that the inequalities~$|n_{11}| > y|n| + \sR(y)$ and~$|n_{21}| > y|n| + \sR(y)$ are incompatible.
Finally, for symmetry reasons, all four inequalities~$|n_{uv}|  > y|n|+\sR(y)$ are incompatible:
this makes our description of the algorithm unambiguous, as announced.

Then, since each node in~$\T$ is~$\cbal{\alpha}$, all children~$n_i$ and grandchildren~$n_{ij}$ of~$n$ satisfy the inequalities~$|n_i| \geqslant (1-\alpha) |n| > |n|/4 > 4$, i.e.,~$|n_i| \geqslant 5$, and~$|n_{ij}| > |n_i|/4 > 1$, i.e.,~$|n_{ij}| \geqslant 2$.
Consequently, these children and grandchildren are all actual tree nodes, and the rotations triggered can always be performed.

When~$|n_{11}| > y |n| + \sR(y) \geqslant \betaa|n|+\sR(\betaa)$, a simple rotation is performed, and
\begin{align}
|n_{11i}| & \leqslant \beta |n_1| \leqslant \alpha \beta |n| \leqslant \betaa|n|+\sR(\betaa) = \betaa|n_1|'+\sR(\betaa); \tag{7.1\textsuperscript{b}} \\
|n_{11}| & \leqslant \alpha |n_1| \leqslant \alpha^2 |n| \leqslant \alphaa |n| + \sR(\alphaa) = \alphaa |n_1|' + \sR(\alphaa); \tag{7.2\textsuperscript{a}} \\
|n_{11}| & \leqslant |n_1|; \tag{7.2\textsuperscript{c}} \\
(1-\alphab)|n_{12}| & \leqslant (1-\alphab)\alpha|n_1| \leqslant (1-\alphab) \alpha^2 |n| \leqslant \alphab(1-\alpha)|n| \leqslant \alphab |n_2|; \tag{7.3\textsuperscript{a}} \\
|n_{12}| & \leqslant 2 \betaa |n|+2\sR(\betaa) - |n_{11}| \leqslant \betaa |n| + \sR(\betaa) = \betaa |n_1|'+\sR(\betaa); \tag{7.3\textsuperscript{b}} \\
(1-\alphab)|n_2| & = (1-\alphab)(|n|-|n_1|) \leqslant  (1-\alphab)((|n_{11}| - \sR(\betaa))/\betaa-|n_1|) \notag \\
& \leqslant (1-\alphab)((\alpha |n_1| - \sR(\betaa))/\betaa-|n_1|) \leqslant \alphab(1-\alpha) |n_1| \leqslant \alphab |n_{12}|; \tag{7.4\textsuperscript{a}} \\
|n_2| & = |n|-|n_1| \leqslant |n|-|n_{11}|/\alpha \notag \\
& \leqslant |n| - (\betaa |n|+\sR(\betaa))/\alpha \leqslant \betaa |n|+\sR(\betaa) = \betaa |n_{12}|'+\sR(\betaa);  \tag{7.4\textsuperscript{b}} \\
|n|' & |n| - |n_{11}| \leqslant (1-\betaa)|n|-\sR(\betaa) \leqslant \alphaa |n| + \sR(\alphaa) = \alphaa |n_1|' + \sR(\alphaa); \tag{7.5\textsuperscript{a}} \\
|n|' & = |n| - |n_{11}| \leqslant (1-\betaa)|n|-\sR(\betaa) \notag \\
& \leqslant (\betaa|n|+\sR(\betaa))/\alpha \leqslant |n_{11}|/\alpha \leqslant |n_1|, \tag{7.5\textsuperscript{c}}
\end{align}
where the inequality~(7.1\textsuperscript{b}) is valid for both children~$n_{11i} = n_{111}$ and~$n_{11i} = n_{112}$ of~$n_{11}$.

Similarly, when~$|n_{12}| > y |n| + \sR(y) \geqslant \betaa |n| + \sR(\betaa)$, a double rotation is performed, and
{\allowdisplaybreaks
\begin{align}
(1-\alphab)|n_{11}| & = (1-\alphab)(|n_1|-|n_{12}|) \leqslant (1-\alphab)(\alpha |n| - \betaa|n| - \sR(\betaa)) \notag \\
& \leqslant \alphab(1-\alpha)(\beta|n|+\sR(\beta)) \leqslant \alphab(1-\alpha)|n_{12}| \leqslant \alphab |n_{121}|; \tag{8.1\textsuperscript{a}} \\
|n_{11}| & \leqslant 2 \betaa |n|+2\sR(\betaa) - |n_{12}| \leqslant \betaa |n| + \sR(\betaa) = \betaa |n_{12}|'+\sR(\betaa); \tag{8.1\textsuperscript{b}} \\
(1-\alphab)|n_{121}| & \leqslant (1-\alphab)\alpha|n_{12}| \leqslant (1-\alphab)\alpha^2|n_1| \leqslant \alphab(1-\alpha)|n_1| \leqslant \alphab|n_{11}|; \tag{8.2\textsuperscript{a}} \\
|n_{121}| & \leqslant \beta |n_1| \leqslant \alpha \beta |n| \leqslant \betaa|n|+\sR(\betaa) = \betaa|n_{12}|'+\sR(\betaa); \tag{8.2\textsuperscript{b}} \\
|n_1|' & = |n_1| - |n_{122}| \leqslant |n_1| - (1-\alpha) |n_{12}| \notag \\
& \leqslant \alpha |n| - (1-\alpha)(\betaa |n| + \sR(\betaa)) \leqslant \alphaa |n| + \sR(\alphaa) = \alphaa |n_{12}|' + \sR(\alphaa); \tag{8.3\textsuperscript{a}} \\
|n_1|' & = |n_1| - |n_{122}| \leqslant |n_1|; \tag{8.3\textsuperscript{c}} \\
(1-\alphab)|n_{122}| & \leqslant (1-\alphab)|n_{12}| \leqslant (1-\alphab) \alpha |n_1| \leqslant (1-\alphab) \alpha^2 |n| \notag \\
& \leqslant \alphab (1-\alpha) |n| \leqslant \alphab |n_2|; \tag{8.4\textsuperscript{a}} \\
|n_{122}| & \leqslant \beta |n_1| \leqslant \alpha \beta |n| \leqslant \betaa|n|+\sR(\betaa) = \betaa|n_{12}|'+\sR(\betaa); \tag{8.4\textsuperscript{b}} \\
|n_2| & = |n|-|n_1| \leqslant |n|-|n_{12}|/\alpha \notag \\
& \leqslant |n|-(\betaa|n|+\sR(\betaa))/\alpha \leqslant \betaa |n|+\sR(\betaa) = \betaa |n_{12}|'+\sR(\betaa); \tag{8.5\textsuperscript{b}} \\
(1-\alphab)|n_2| & \leqslant (1-\alphab)(|n|-(\betaa|n|+\sR(\betaa))/\alpha) \notag \\
& \leqslant \alphab(1-\alpha)(\betaa|n|+\sR(\betaa)) \leqslant \alphab(1-\alpha) |n_{12}| \leqslant \alphab |n_{122}|; \tag{8.5\textsuperscript{a}} \\
|n|' & = |n|-|n_1|+|n_{122}| \leqslant |n| - (1-\beta)|n_1| \leqslant |n| - (1-\beta) |n_{12}| / \alpha \notag \\
& \leqslant |n| - (1-\beta) (\betaa |n| + \sR(\betaa)) / \alpha \leqslant \alphaa |n| + \sR(\alphaa) = \alphaa |n_{12}|' + \sR(\alphaa); \tag{8.6\textsuperscript{a}} \\
|n|' & = |n|-|n_1|+|n_{122}| \leqslant |n| - (1-\beta)|n_1| < u |n| - (1-\beta)|n_1| + v \notag \\
& < u (|n_{12}|-\sR(\betaa))/\betaa - (1-\beta)|n_1| + v \notag \\
& < u (\alpha|n_1|-\sR(\betaa))/\betaa - (1-\beta)|n_1| + v = |n_1| + 1. \tag{8.6\textsuperscript{c}}
\end{align}}

This means, in both cases, that the root of~$\T'$ is~$\cbal{\alphaa}$ and robustly~$\gcbal{\beta}$, that the weights of its children cannot exceed~$|n_1|$, and that both left and right sub-trees of~$\T'$ are~$\GCBB{\alphab,1}{\alpha,\beta}$-trees.

Finally, the cases where~$|n_{21}| > y |n| + \sR(y)$ or~$|n_{22}| > y |n| + \sR(y)$ are symmetrical, and thus the conclusion of Lemma~\ref{lem:RB-balancing} is valid in these cases too.
\end{proof}

With the help of Lemmas~\ref{lem:RA-balancing} and~\ref{lem:RB-balancing}, we can now prove Proposition~\ref{pro:RAB-balancing} itself, which we first recall.

\rabbalancing*

\begin{proof}
If the root of~$\T$ is robustly~$(\alpha,\beta)$-balanced, the~$\sR\sA\sB$-balancing algorithm does nothing.

If this root is robustly~$\cbal{\alpha}$ but not~$\gcbal{\beta}$, the~$\sB(\beta,\sR(\beta))$-balancing algorithm is called; it results in a tree~$\T^{(3)}$, whose root will be denoted by~$r$.
Nodes affected by this call are~$r$ and one or two of its children~$r_1$ and~$r_2$.
Lemma~\ref{lem:RB-balancing} states that~$\balc(\T') \leqslant \balc(\T)$ and that~$r$ is robustly~$(\alphaa,\betaa)$-balanced, whereas the sub-trees rooted at~$r_1$ and~$r_2$ are only guaranteed to be~$\GCBB{\alphab,1}{\alpha,\beta}$-trees.
The~$\sA\sB(\alphaa,\alphaa,\betaa,\betaa)$-balancing algorithm is then applied to both sub-trees, which transforms them into~$\GCB{\alpha,\beta}$-trees whose roots and affected nodes are~$(\alphaa,\betaa)$-balanced, without damaging the \gcbal{} of~$r$.
Hence, the conclusion of Proposition~\ref{pro:RAB-balancing} is valid in this case.

The situation is similar when the root of~$\T$ is not robustly~$\cbal{\alpha}$:
the~$\sA(\alpha,\alphaa,\sR(\alpha),\sR(\alphaa))$-balancing algorithm yields a tree~$\T^{(1)}$ whose root~$r$ is robustly~$\cbal{\alphaa}$ and whose children~$r_1$ and~$r_2$ are~$\cbal{\alphab}$;
subsequently calling the~$\sA\sB(\alphaa,\alphaa,\betaa,\betaa)$-balancing algorithm on sub-trees rooted at~$r_1$ and~$r_2$ results in a~$\GCB{\alpha,\beta}$-tree~$\T^{(2)}$ whose root is robustly~$\cbal{\alphaa}$ and whose affected nodes are~$(\alphaa,\betaa)$-balanced.
Then, additional calls to the~$\sB(\betaa,\sR(\betaa))$- and~$\sA\sB(\alphaa,\alphaa,\betaa,\betaa)$-balancing algorithms yield, like in the previous paragraph, a~$\GCB{\alpha,\beta}$-tree~$\T'$ whose root is robustly~$(\alphaa,\betaa)$-balanced and whose affected nodes are~$(\alphaa,\betaa)$-balanced.

Finally, if the root of~$\T$ it is not~$\cbal{u}$, we will call the~$\sA(u,\alphaa,0,0)$-algorithm, obtaining a~$\GCB{\alpha,\beta}$-tree~$\T^{(1)}$, whose root will be denoted by~$r$.
Nodes affected by this call are~$r$ and one or two of its children~$r_1$ and~$r_2$.
The~$\sB(\betaa,0)$-balancing algorithm is then applied to the~$\GCBB{\alpha,1}{\alpha,\beta}$-trees rooted at~$r_1$ and~$r_2$;
as a result,~$\T^{(2)}$ is a~$\GCBB{\alpha,1}{\alpha,\beta}$-tree rooted at~$r$, to which the~$\sB(\betaa,0)$-balancing algorithm is applied once more.
No tree to which the~$\sA$- and~$\sB$-balancing algorithms were applied saw its~$\cbal{}$ increase, and thus~$\balc(\T') \leqslant \balc(\T)$. 

Moreover, each node~$m$ affected by the~$\sA\sB$-balancing algorithm is either affected by some call to the~$\sB$-balancing algorithm or the root of some tree that the~$\sB$-balancing algorithm did not modify; the latter case may only concern nodes~$r$,~$r_1$ and~$r_2$.
In both cases,~$m$ ends up being~$(\alphaa,\betaa)$-balanced, which completes the proof.
\end{proof}

This procedure being valid, it is then time to prove the correctness of our bottom-up updating algorithm.

\tdbalancing*

\begin{proof}
First note that step~2 of Algorithm~\ref{alg:TD-updating}, which consists in updating node weights, might as well be performed in a new top-down pass occurring after the tree has been completely rebalanced; indeed, once such a weight is updated, it will never be read again.
Hence, we safely omit this step, and look at the tree~$\T'$ obtained from~$\T$ just before applying Algorithm~\ref{alg:TD-updating} to a sub-tree of weight~$30$ or less.

In this context, we will prove the following claim by induction on~$|\T|$.
\begin{quote}
Let~$n_0$ be the root of~$\T'$, and let~$n_k$ be the highest node of~$\T'$ of weight~$|n_k| \leqslant 30$ such that the leaf~$s$ to be inserted or deleted descends from~$n_k$.
Let~$n_0, n_1, \ldots, n_k$ be the path from~$n_0$ to~$n_k$ in~$\T'$.
The tree~$\T'$ is a~$\GCB{\alpha,\beta}$-tree such that~$\balc(\T') \leqslant \balc(\T)$, and each of the nodes~$n_0, n_1, \ldots, n_{k-1}$ is robustly~$(\alpha,\beta)$-balanced.
\end{quote}

This claim is vacuously true when~$|\T| = |\T'| \leqslant 30$.
If~$|\T| \geqslant 31$, we start by applying step~1 of Algorithm~\ref{alg:TD-updating} to~$\T$.
According to Proposition~\ref{pro:RAB-balancing}, we thereby transform~$\T$ it into a~$\RGCB{\alpha,\beta}$-tree~$\T_1$ such that~$\balc(\T_1) \leqslant \balc(\T)$.
Our next move consists in recursively updating a sub-tree of~$\T_1$, say~$\T_2$: by induction hypothesis, this transforms~$\T_2$ into a~$\GCB{\alpha,\beta}$-tree~$\T'_2$ such that~$\balc(\T'_2) \leqslant \balc(\T_2)$ and in which there exists a path from the root to a node~$n_k$ such that~$|n_k| \leqslant 30$ from which~$s$ descends, and formed of robustly~$(\alpha,\beta)$-balanced nodes (except the node~$n_k$ itself).

Transforming~$\T_2$ into~$\T'_2$ also transforms~$\T_1$ into~$\T'$.
Since~$\balc(\T') = \balc(\T_1) \leqslant \balc(\T)$ and~$\balc(\T'_2) \leqslant \balc(\T_2)$, we also have~$\balgc(\T') \leqslant \balgc(\T_1)$, which proves that the root of~$\T_1$ remained robustly~$(\alpha,\beta)$-balanced during this transformation.
Consequently, the claim is valid for~$\T$, which completes the induction.

We are now ready to prove the correctness of Algorithm~\ref{alg:TD-updating}.
First, the algorithm transforms~$\T$ into the tree~$\T'$ described above.
Then, Algorithm~\ref{alg:TD-updating} is applied to a sub-tree~$\T''$ of~$\T'$, of weight~$30$ or less, thereby correctly inserting a leaf~$s$ into~$\T''$ deleting a leaf from~$\T''$, or just rebalancing~$\T''$ (if the update was redundant), and updating the weights of the nodes~$n_0, n_1, \ldots, n_{k-1}$ (i.e., applying step~2 of Algorithm~\ref{alg:TD-updating}, which we had postponed).
Doing so transforms the trees~$\T'$ and~$\T''$ into new trees~$\T'_\star$ and~$\T''_\star$, respectively.

Regardless of whether the update consisted in inserting~$s$, deleting~$s$ or being redundant, each of the nodes~$n_0, n_1, \ldots, n_{k-2}$ remained~$(\alpha,\beta)$-balanced, and~$n_{k-1}$, if it exists, remained~$\cbal{\alpha}$.
Furthermore, the tree~$\T''_\star$ itself is perfectly balanced.
Hence, it just remains to prove that~$n_{k-1}$, if it exists, remained~$\gcbal{\beta}$.
To do so, we rely on the inequality
\begin{align}
\alpha(k-1)+2 < 2 \beta(k-1), \tag{9}
\end{align}
which is valid whenever~$k \geqslant 10$.
Indeed, let~$|m|$ be the weight of a node~$m$ in~$\T'$, and let~$|m|_\star$ be its weight in~$\T'_\star$:
each grandchild~$m$ of~$n_{k-1}$ in~$\T'_\star$ is of weight
\[|m|_\star \leqslant (|n_k|_\star+1)/2 \leqslant (\alpha(|n_{k-1}|-1) + 2)/2 \leqslant \beta(|n_{k-1}|-1) \leqslant \beta(|n_{k-1}|_\star)\]
if~$m$ belongs to~$\T''_\star$, and of weight
\[|m|_\star = |m| \leqslant \beta |n_{k-1}|+\sR(\beta) \leqslant \beta(|n_{k-1}|-1) \leqslant \beta |n_{k-1}|_\star\]
otherwise.
This means that~$n_{k-1}$ remained~$\gcbal{\beta}$, which completes the proof.
\end{proof}

\section{Proofs of Section~\ref{sec:complexity}}

Our complexity proof starts with the following (crude) estimation of the magnitude at which balances can increase, which borrowed from~\cite{BM80}.

\begin{lemma}\label{lem:balance-increase}
Let~$r$ be the root of a binary search tree~$\T$, which is transformed into a new tree~$\T'$ by inserting a leaf in~$\T$ or deleting a leaf~$s \neq r$ from~$\T$, and let~$|r|$ be the weight of~$r$ in~$\T$.
We have~$\balc(\T',r) \leqslant \balc(\T,r)+2/|r|$ and~$\balgc(\T',r) \leqslant \balgc(\T,r)+2/|r|$.
\end{lemma}

\begin{proof}
Since no \textsf{C}-balance nor \textsf{GC}-balance exceeds~$1$, the desired result is immediate when~$|r| \leqslant 2$.
Hence, we assume that~$|r| \geqslant 3$, and we consider a descendant~$m$ of~$r$ in~$\T'$.
If~$m$ is the node that was just inserted in~$\T$, it occupies the place of a previously empty tree of weight~$1$, and thus we abusively see it as a former node of weight~$1$ in~$\T$.

Let~$|m|$,~$|r|'$ and~$|m|'$ denote the weight of~$m$ in~$\T$ and the weights of~$r$ and~$m$ in~$\T'$.
Let also~$z = |m|/|r|$: it suffices to prove that~$|m|'/|r|' \leqslant z + 2/|r|$.
And indeed, if~$n$ was inserted in~$\T$,
\[\frac{|m|'}{|r|'} = \frac{|m|'}{|r|+1} \leqslant \frac{|m|+1}{|r|+1} = z + \frac{1-z}{|r|+1} \leqslant z + \frac{2}{|r|},\]
whereas, if~$n$ was deleted from~$\T$,
\[\frac{|m|'}{|r|'} = \frac{|m|'}{|r|-1} \leqslant \frac{|m|}{|r|-1} = z + \frac{z}{|r|-1} \leqslant z + \frac{3z/2}{|r|} \leqslant z + \frac{2}{|r|}.\qedhere\]
\end{proof}

\bucomplexity*

\begin{proof}
Let~$\U_0,\U_1,\ldots,\U_\ell$ be the binary search trees obtained during our gradual transformation of~$\T_0 = \U_0$ into~$\T_k = \U_\ell$: each tree~$\U_{i+1}$ is obtained by inserting a leaf in~$\U_i$ or deleting a leaf from~$\U_i$, which happens~$k$ times, by applying the~$\sA\sB$,~$\sR\sA\sB$-balancing algorithm to a sub-tree of~$\U_i$, or by perfectly rebalancing a sub-tree with up to~$29$ nodes (which we call \emph{base-case balancing}).
Calls to the~$\sA\sB$-balancing algorithm triggered inside a call to the~$\sR\sA\sB$-balancing algorithm are not counted.

Below, the weight of a node~$n$ of~$\U_i$ is denoted by~$|n|_i$, and the sub-tree of~$\U_i$ rooted at~$n$ is denoted by~$\U_i(n)$.
We also set~$\delta = \min\{\alpha - \alphaa,\beta - \betaa,1/62\}$.
Each node~$n$ of~$\U_i$ is given a real-valued counter~$\scc_i(n)$ that we initialise and update as follows:
\begin{enumerate}[itemsep=0.5pt,topsep=0.5pt]
\item when a node~$n$ is inserted in a tree~$\U_i$, we set~$\scc_{i+1}(n) = 0$;

\item when a node is inserted in~$\U_i(n)$ or deleted from~$\U_i(n)$, we set~$\scc_{i+1}(n) = \scc_i(n) + 2/|n|_i$;

\item when the~$\sA\sB$- or~$\sR\sA\sB$-balancing algorithm affects a node~$n$, we set~$\scc_{i+1}(n) = 0$;

\item in all other cases, we simply set~$\scc_{i+1}(n) = \scc_i(n)$.
\end{enumerate}

An immediate induction on~$i$ proves that~$\scc_i(n) \geqslant 0$ for all nodes~$n$ of~$\U_i$.
We also prove by induction on~$i$ that each node~$n$ of~$\U_i$ is~$\cbal{(\alphaa+\scc_i(n))}$, and that~$n$ is~$\gcbal{(\betaa+\scc_i(n))}$ if~$\scc_i(n) < \delta$:
\begin{enumerate}[itemsep=0.5pt,topsep=0.5pt]
\item When~$n$ is inserted in~$\U_i$, it is a leaf; it follows that~$\balc(\U_{i+1},n) = 1/2 \leqslant \alphaa$ and that~$\balgc(\U_{i+1},n) = 1/4 \leqslant \betaa$.

\item When a node is inserted in~$\U_i(n)$ or deleted from~$\U_i(n)$, Lemma~\ref{lem:balance-increase} and the induction hypothesis prove that
\[\balc(\U_{i+1},n) \leqslant \balc(\U_i,n) + 2/|n|_i \leqslant \alphaa+\scc_i(n) + 2/|n|_i = \alphaa+\scc_{i+1}(n).\]
Furthermore, if~$\scc_{i+1}(n) < \delta$,
\[\balgc(\U_{i+1},n) \leqslant \balgc(\U_i,n) + 2/|n|_i \leqslant \betaa+\scc_i(n) + 2/|n|_i = \betaa+\scc_{i+1}(n).\]

\item When the~$\sA\sB$- or~$\sR\sA\sB$-balancing algorithm affects~$n$, Propositions~\ref{pro:AB-balancing} and~\ref{pro:RAB-balancing} ensure that~$\balc(\U_{i+1},n) \leqslant \alphaa$ and~$\balgc(\U_{i+1},n) \leqslant \betaa$.

\item When a base-case balancing affects~$n$, since~$|n|_{i+1} \geqslant 2$, we have~$\balc(\U_{i+1},n) = 1/2 \leqslant \alphaa$ if~$|n|_{i+1}$ is even, and~$\balc(\U_{i+1},n) = (|n|_i+1)/(2|n|_i) \leqslant 2/3 \leqslant \alphaa$ if~$|n|_{i+1}$ is odd.
Thus, we also have~$\balgc(\U_{i+1},n) \leqslant 4/9 \leqslant \betaa$.

\item In other cases,~$\balc(\U_{i+1},n) = \balc(\U_i,n) \leqslant \scc_i(n) = \scc_{i+1}(n)$ and, if~$\scc_{i+1}(n) < \delta$, we have~$\scc_i(n) < \delta$ and~$\balgc(\U_{i+1},n) = \balgc(\U_i,n) \leqslant \scc_i(n) = \scc_{i+1}(n)$.
\end{enumerate}

We will then evaluate the sum~$\sC_i$ of all real numbers values~$\scc_i(n)$ where~$n$ ranges of the nodes of~$\U_i$:
\begin{enumerate}[itemsep=0.5pt,topsep=0.5pt]
\item When a leaf~$s$ is inserted or deleted,~$\sC_{i+1}$ is obtained by either adding a zero counter value (in case of insertion), deleting a non-negative counter value (in case of deletion) or doing nothing yet (in case the update is redundant), and in increasing the counter values of all ancestors of~$n$.
Since the~$h$\textsuperscript{th} ancestor of~$s$ in~$\U_i$ has a weight~$|s^{(h)}|_i \geqslant \alpha^{-h}$, we know that~$\scc_{i+1}(s^{(h)}) \leqslant \scc_i(s^{(h)}) + 2\alpha^h$.
It follows that
\[\sC_{i+1} \leqslant \sC_i + \sum_{h \geqslant 0} 4 \alpha^h = \sC_i + 4/(1-\alpha) \leqslant \sC_i + 16.\]

\item When the~$\sA\sB$- or~$\sR\sA\sB$-balancing algorithm, or a base-case balancing, affects node in a sub-tree~$\U_i(n)$, we have~$\scc_{i+1}(m) = 0 \leqslant \scc_i(m)$ for all affected nodes~$m$.
This means that~$\sC_{i+1} \leqslant \sC_i$.

If, furthermore,~$|n| \geqslant 2/\delta \geqslant 31$, the node~$n$ failed to be robustly~$(\alpha,\beta)$-balanced in~$\U_i$.
This means that~$\balc(\U_i,n) \geqslant \alpha - \sR(\alpha) / |n| \geqslant \alphaa + \delta/2$ or~$\balgc(\U_i,n) \geqslant \beta - \sR(\beta) / |n| \geqslant \betaa + \delta/2$, and therefore that~$\sC_{i+1} \leqslant \sC_i - \delta/2$.
\end{enumerate}

Now, let us count calls to the~$\sA\sB$-,~$\sR\sA\sB$- or base-case balancing algorithms when transforming a tree~$\T_u$ into the tree~$\T_{u+1}$ by using Algorithm~\ref{alg:BU-rebalancing} or~\ref{alg:TD-updating}.
The base-case algorithm was called at most once, triggering~$\O(1)$ rotations.
In addition, the~$\sA\sB$- or~$\sR\sA\sB$-balancing algorithms were applied to sub-trees of weights~$w_1 > w_2 > \cdots > w_j$.
No more than~$2/\delta$ such calls concerned trees of weight~$w < 2/\delta$, each other call decreasing the value of the sum~$\sC_i$ by at least~$\delta/2$;
calls are said to be \emph{small} if they belong to the former category, or \emph{large} if they belong to the latter one.

Consequently, during the entire process of transforming the tree~$\U_0 = \T_0$ into~$\U_\ell = \T_k$, let~$\sB\sC$,~$\sS$ and~$\sL$ the number of base-case, small and large rebalancing operations, respectively.
Each one of them triggers~$\O(1)$ write operations, and we just proved that~$\sB\sC \leqslant k$,~$\sS \leqslant 2 k / \delta$ and~$0 \leqslant \sC_\ell \leqslant \sC_0 + 16 k - \mathsf{L} \delta / 2$, which means that~$\sL \leqslant 32 k / \delta$.
Therefore,~$\O(\sB\sC + \sS + \sL + k) = \O(k / \delta)$ rotations were triggered.

Finally, with a computer algebra system, we observe that
\begin{align}
\delta \geqslant \min\{\alpha-1/\sqrt{2},\beta - \scrB(\alpha),\alpha^2-\beta\}/12. \tag{10.1}
\end{align}
The desired result follows in case~$\beta < \alpha^2$.

The situation is in fact slightly simpler when~$\beta = \alpha^2$.
Indeed, in that case, we observe, again with a computer algebra system, that~$\alpha(\alpha k + \sR(\alpha)) \leqslant \beta k + \sR(\beta)$
for all integers~$k$.
Hence, a node that is~$\cbal{\alpha}$ is necessarily~$\gcbal{\beta}$, and a node that is robustly~$\cbal{\alpha}$ is necessarily robustly~$\gcbal{\beta}$ too.
In other words, when~$\beta = \alpha^2$, steps~3 of both algorithms~\ref{alg:AB} and~\ref{alg:RAB} never lead to any rotation.
Consequently, in the above analysis, we may simply define~$\delta$ as~$\delta = \min\{\alpha-\alphaa,1/62\}$, observing this time that~$\delta \geqslant \alpha-1/\sqrt{2}$.
This completes the proof of Theorem~\ref{thm:TD-balancing}.
\end{proof}
\end{document}